\documentclass{lmcs}
\pdfoutput=1
\usepackage[utf8]{inputenc}

\usepackage{lastpage}
\lmcsdoi{21}{4}{8}
\lmcsheading{}{\pageref{LastPage}}{}{}%
{Dec.~13,~2024}{Oct.~15,~2025}{}

 %TODO mandatory; please add comma-separated list of keywords

\usepackage{amsmath}
\usepackage{amssymb}
\usepackage{xspace}
\usepackage{ebproof}
\usepackage{hyperref}
\usepackage[capitalize,nameinlink]{cleveref}
\usepackage{multirow}
\usepackage{booktabs}
\theoremstyle{plain}

\newcommand{\acc}{\ensuremath{\widehat{q}}\xspace}

\newcommand{\RHL}{\ensuremath{R_{\mathit{HL}}}\xspace}

\newcommand{\RM}{\ensuremath{R^{\M}}\xspace}
\newcommand{\RMbin}{\ensuremath{R^{\M}_{\mathit{bin}}}\xspace}

\newcommand{\EL}{\ensuremath{\mathcal{EL}}\xspace}
\newcommand{\N}{\ensuremath{\mathbb{N}}\xspace}
\newcommand{\R}{\ensuremath{\mathbb{R}}\xspace}

\newcommand{\M}{\ensuremath{\mathcal{M}}\xspace}

\newcommand{\dto}[1]{\mathrel{\displaystyle\mathop{\to}^{#1}}}

\newcommand{\too}{\dto{*}}

\newcommand{\CReins}{\ensuremath{\mathsf{R1}}\xspace}
\newcommand{\CRzwei}{\ensuremath{\mathsf{R2}}\xspace}
\newcommand{\CRdrei}{\ensuremath{\mathsf{R3}}\xspace}

\newcommand{\gci}[2]{{\sqsubseteq}(#1,#2)}
\newcommand{\er}[2]{E_#1(#2)}
\newcommand{\wnull}{w_{\mathit{min}}}
\newcommand{\wmax}{w_{\mathit{max}}}
\newcommand{\dep}{\mathit{dp}}
\newcommand{\scomp}{size compatible\xspace}
\newcommand{\scompa}{size compatibility\xspace}
\newcommand{\nusz}[1]{|#1|_{\mathit{nu}}}

\crefname{defi}{Definition}{Definitions}
\crefname{exa}{Example}{Examples}
\crefname{section}{Section}{Sections}
\crefname{subsection}{Section}{Sections}
\crefname{subsubsection}{Section}{Sections}
\crefname{prop}{Proposition}{Propositions}
\crefname{lem}{Lemma}{Lemmas}
\crefname{thm}{Theorem}{Theorems}
\crefname{cor}{Corollary}{Corollaries}

\begin{document}

\title[Term Reachability Problems]{Small Term Reachability and Related Problems\texorpdfstring{\\}{} for Terminating Term Rewriting Systems}
\titlecomment{{\lsuper*}This is a revised and extended journal version of our earlier
  conference paper \cite{FSCD2024}.} 
\thanks{\emph{Franz
  Baader}:
  Partially supported by DFG, Grant 389792660, within TRR 248 ``Center for Perspicuous
  Computing'', and by the German Federal Ministry of Education and Research (BMBF,
  SCADS22B) and the Saxon State Ministry for Science, Culture and Tourism (SMWK) by
  funding the competence center for Big Data and AI ``ScaDS.AI Dresden/Leipzig''.\\
  \emph{Jürgen Giesl}: Partially supported by DFG, Grant 235950644 (Project GI 274/6-2).}
\author[F.~Baader]{Franz Baader\lmcsorcid{0000-0002-4049-221X}}[a]
\author[J.~Giesl]{J\"urgen Giesl\lmcsorcid{0000-0003-0283-8520}}[b]
% affiliation 1 (automatically numbered a)
\address{Theoretical Computer Science, TU Dresden, Germany \and SCADS.AI
  Dresden/Leipzig,
  Germany}
\email{franz.baader@tu-dresden.de}
% affiliation 2 (automatically numbered b)
\address{RWTH Aachen University, Aachen, Germany}	%optional
\email{giesl@informatik.rwth-aachen.de}

%%%%%%%%%%%%%%%%%%%%%%%%%%%%%%%%%%%%%%%%%%%%%%%%%%%%%%

\begin{abstract}
Motivated by an application where we try to make proofs for Description Logic inferences smaller by rewriting, we
consider the following decision problem, which we call the small term reachability problem: 
given a term rewriting system $R$, a term $s$, and a natural number $n$, decide whether there is a term $t$ of size $\leq n$ reachable
from $s$ using the rules of $R$. We investigate the complexity of this problem depending on how termination of $R$ can be established.
We show that the problem is in general NP-complete for length-reducing term rewriting systems. Its
complexity increases to N2ExpTime-complete (NExpTime-complete) if termination
is proved using a (linear) polynomial order and
to PSpace-complete for systems whose
termination can be shown using a restricted class of Knuth-Bendix orders.
Confluence reduces the complexity to P for the length-reducing case, but has no effect on
the worst-case complexity in the other two cases.
Finally, we consider the large term reachability problem, a variant of the problem where we are interested in reachability of a term of size $\geq n$.
It turns out that this seemingly innocuous modification in some cases changes the complexity of the problem, which may also become dependent on 
whether the number $n$ is is represented in unary or binary encoding, whereas this makes
no difference for the complexity of the small term reachability problem.
\end{abstract}

\maketitle

\section{Introduction}

Term rewriting~\cite{BaNi98,books/daglib/0008995} is a well-investigated formalism, which can be used both for computation and deduction.
A term rewriting system $R$ consists of rules, which describe how a term $s$ can be transformed into a new term $t$, in which case one writes $s\to_R t$.
In the computation setting, where term rewriting is akin to functional programming~\cite{DBLP:journals/toplas/GieslRSST11}, 
a given term (the input) is iteratively rewritten into a normal form (the output), which
is a term that cannot be rewritten further. Termination of $R$ prevents infinite rewrite sequences, and thus ensures that a normal form can always be reached,
whereas confluence guarantees that the output is unique, despite the nondeterminism inherent to the rewriting process (which rule to apply when and where).
In the deduction setting, which is, e.g., relevant for first-order theorem proving with equality~\cite{DBLP:books/el/RV01/NieuwenhuisR01}, 
one is interested in whether a term $s$ can be rewritten into a term
$t$ by iteratively applying the rules of $R$ in both directions. If $R$ is confluent and terminating, this problem can be solved by computing normal forms of
$s$ and $t$, and then checking whether they are equal. In the present paper, we want to employ rewriting for a different purpose: given a term $s$, we are
interested in finding a term $t$ of minimal size that can be reached from $s$ by rewriting (written $s\too_R t$), but this term need not be in normal form.
To assess the complexity of this computation problem, we investigate the corresponding decision problem: given a term rewriting system $R$, a term $s$, and
a natural number $n$, decide whether there is a term $t$ of size $\leq n$ such that $s\too_R t$. We call this the \emph{small term reachability problem}.

Our interest in this problem stems from the work on finding small proofs~\cite{DBLP:conf/lpar/AlrabbaaBBKK20,DBLP:conf/cade/AlrabbaaBBKK21} 
for Description Logic (DL) inferences~\cite{DBLP:books/daglib/0041477}, 
which are then visualized in an interactive explanation tool~\cite{DBLP:conf/cade/AlrabbaaBBDKM22}. 
For the DL $\mathcal{EL}$~\cite{DBLP:conf/ijcai/BaaderBL05}, we employ the highly-efficient reasoner ELK~\cite{DBLP:journals/jar/KazakovKS14} 
to compute proofs. However, the proof calculus employed by ELK is rather fine-grained, and thus produces relatively large proofs. Our idea was thus 
to generate smaller proofs by rewriting several proof steps into a single step. 
As a (simplified) example, consider the three proof rules in \Cref{proof:fig}.
\begin{figure}[t!]
\parbox[c]{\textwidth}{%
\hfill
\[
\CReins\ \ \frac{A \sqsubseteq B \ \ B \sqsubseteq C}{A \sqsubseteq C} \qquad 
\CRzwei\ \ \frac{A \sqsubseteq B}{\exists r . A \sqsubseteq \exists r . B}
\qquad 
\CRdrei\ \ \frac{A \sqsubseteq \exists r . A_1 \ \ A_1 \sqsubseteq B_1 \ \  \exists r . B_1 \sqsubseteq B}{A \sqsubseteq B}
\]
}
\caption{Three proof rules for \EL.}
\label{proof:fig}
\end{figure} 
\begin{figure}[t!]
\parbox[c]{\textwidth}{%
\[
    \begin{prooftree}
        \hypo{A\sqsubseteq \exists r . A_1}
        \hypo{A_1 \sqsubseteq B_1}
        \infer[left label=\CRzwei]1{\exists r . A_1\sqsubseteq \exists r . B_1}
        \hypo{\exists r . B_1 \sqsubseteq B}
        \infer[left label=\CReins]2{\exists r . A_1\sqsubseteq B}
        \infer[left label=\CReins]2{A\sqsubseteq B}
    \end{prooftree}
\]
}
\caption{Proof of the conclusion of \CRdrei from its hypotheses using \CReins and \CRzwei.}
\label{proof:fig2}
\end{figure}%
It is easy to see that one needs one application of \CRzwei followed by two of \CReins to produce the same consequence as a single application of \CRdrei
(see \Cref{proof:fig2}).
Thus, if one looks for patterns in a proof that use \CReins and \CRzwei in this way, and replaces them by the corresponding applications of \CRdrei, then
one can reduce the size of a given proof. Given finitely many such proof rewriting rules and a proof, the question is then how to use the rules to
rewrite the given proof into one of minimal size. Since tree-shaped proofs as well as DL concept descriptions can be represented as terms,
this question can be seen as an instance of the small term reachability problem\linebreak introduced
above.

For example, the proof consisting only of one application of \CRdrei (as depicted in
\Cref{proof:fig}) can be written as the term
$$
t_1 := \CRdrei(\gci{A}{B},\gci{A}{\er{r}{A_1}},\gci{A_1}{B_1},\gci{\er{r}{B_1}}{B}),
$$
where we use \CRdrei as a 4-ary function symbol whose first argument is the consequence of applying the corresponding rule and the other
arguments are the hypotheses (or more generally, the proof terms producing these
hypotheses).
In addition, $\sqsubseteq$ is used as a binary
function symbol, and $E_r$ as a unary function symbol.
The proof in \Cref{proof:fig2} can be represented as the proof term
$$
\begin{array}{l@{\,}l@{\,}l}
t_2 := \CReins(\gci{A}{B},&\gci{A}{\er{r}{A_1}},\\
                   &\CReins(\gci{\er{r}{A_1}}{B},&\CRzwei(\gci{\er{r}{A_1}}{\er{r}{B_1}},\gci{A_1}{B_1}),\\
                   &&\gci{\er{r}{B_1}}{B})).
\end{array}
$$
Before we can use these terms in a rewrite rule of the form $t_2\rightarrow t_1$, we must make two changes. First, note that
the rules in \Cref{proof:fig} are actually rule schemata, where $A,B,C,A_1,B_1$ are placeholders for \EL concepts. Thus, we
must view them as variables that may be replaced be terms representing \EL concepts in the terms $t_1$ and $t_2$.
In addition, when applying such a rewrite rule within a larger proof, the hypotheses in these terms may also be derived by a subproof,
and the formulation of the rewrite rule must take this possibility into account. This means, for instance, that $\gci{A}{\er{r}{A_1}}$ in both $t_1$
and $t_2$ is replaced with $\CReins(\gci{A}{\er{r}{A_1}},Z_1,Z_2)$, where $Z_1$ and $Z_2$ are variables that may be substituted by proof terms,
and analogous variants must be considered for the other rules. We refrain from giving more details on how to express proof rewrite rules as term rewrite rules
since the main topic of this paper is the investigation of the small term reachability problem in the setting of term rewriting systems.
However, the sketch in this paragraph shows that proof rewriting can indeed be expressed as term rewriting.

In the following we investigate the complexity of the small term reachability problem on the general level of term rewriting systems (TRSs).
It turns out that this complexity depends on how termination of the given TRS can be shown. It should be noted that, in our complexity results,
we assume the TRS $R$ to be fixed, and only the term $s$ and the number $n$ are the variable part of the input. Thus, in the subsequent summary of our 
results, we say that the problem is \emph{in general} complete for a complexity class $\mathcal{K}$ if, for every fixed
TRS falling into the respective category, the problem is in $\mathcal{K}$, and there is a TRS belonging to this category for which the problem is also
hard for $\mathcal{K}$. The paper contains the following main contributions (see
\Cref{Overview} for an overview):

\begin{table}[t]
  \begin{center}
   \scalebox{0.83}{
      \makebox[\textwidth][c]{
        \setlength{\tabcolsep}{2pt}
        \begin{tabular}{c@{\qquad}ccc@{\qquad}cc}
          \toprule
          \multirow{2}{*}{class of TRS}& \multicolumn{2}{c}{small term reachability} &&
          \multicolumn{2}{c}{large term reachability}\\
          \cmidrule{2-3} \cmidrule{5-6}
 & upper bound & lower bound && upper bound & lower bound\\
                    \toprule                  \addlinespace  
                    \multirow{2}{*}{length-reducing} & NP & NP && linear & \\
                    &\multicolumn{2}{c}{(\Cref{completeness:main:thm})}&&(\Cref{large:length-reducing})&\\
                 \addlinespace    \midrule  \addlinespace  
 length-reducing & P &  &&  linear &\\
 \& confluent & (\Cref{confluence:prop}) &&&(\Cref{large:length-reducing})&\\
  \addlinespace \midrule  \addlinespace  
 terminating with& PSpace & PSpace &&  PSpace & PSpace\\
  KBO without&
  \multicolumn{2}{c}{(\Cref{kbo:thm})}&&\multicolumn{2}{c}{(\Cref{largeTerm-KBO})}\\
 special symbol &
 \multicolumn{2}{c}{also for confluence}&&\multicolumn{2}{c}{also for confluence}\\
  \addlinespace \midrule  \addlinespace  
  & & &&  NExpTime & NExpTime\\
  & & && (for binary encoding)&(for binary encoding)\\
  terminating with & NExpTime & NExpTime &&    \multicolumn{2}{c}{(\Cref{large:term:lin:po:thm})}\\
   size compatible &  \multicolumn{2}{c}{(\Cref{Hardness Linear Polynomials})} &&
   \multicolumn{2}{c}{also for confluence}\\\addlinespace[-.05cm]\cmidrule{5-6}\addlinespace[-.08cm]
   linear pol.\ order &  \multicolumn{2}{c}{also for confluence} && PSpace&\\
  &&&&(for unary encoding)&\\
   &&&&(\Cref{general:upper:thm})&\\
  \addlinespace \midrule  \addlinespace  
  & & &&  ExpSpace & \\
   & & &&  (for binary  encoding)&\\
   terminating with & N2ExpTime & N2ExpTime && (\Cref{general:upper:thm})&\\\addlinespace[-.05cm]\cmidrule{5-6}\addlinespace[-.08cm]
   size compatible & \multicolumn{2}{c}{(\Cref{completeness:main:thm:poly})} &&PSpace & \\
   pol.\ order & \multicolumn{2}{c}{also for confluence} && (for unary encoding)&\\
   & & && (\Cref{general:upper:thm})&\\
  \addlinespace \midrule  \addlinespace  
   & & &&  ExpSpace & \\
   & & &&  (for binary  encoding)&\\
  terminating  & \multicolumn{2}{c}{decidable} && (\Cref{general:upper:thm})&\\\addlinespace[-.05cm]\cmidrule{5-6}\addlinespace[-.08cm]
    & \multicolumn{2}{c}{(\Cref{decidability-prop})} &&PSpace & \\
  & & && (for unary encoding)&\\
   & & && (\Cref{general:upper:thm})&\\
 \bottomrule
       \end{tabular}}}
    \caption{Overview on our complexity results}
  \label{Overview}
  \end{center}
\end{table}

\paragraph*{1.\ Small term reachability for length-reducing TRSs}

If the introduced rewrite rules are \emph{length-reducing}, i.e., each rewrite step
decreases the size of the term (proof), like the rule in our example, 
then termination of all rewrite sequences is guaranteed. In general, it may nevertheless be the case that one can generate two normal forms of different sizes. 
Confluence prevents this situation, i.e., then it is sufficient to generate only one rewrite sequence to produce a term (proof) of minimal size. 
In \Cref{length:red:sect} we show that the small term reachability problem for length-reducing term rewriting systems is in general NP-complete,
but becomes solvable in polynomial time 
for confluent systems.

\paragraph*{2.\ Small term reachability for  TRSs whose termination is shown
  by polynomial orders}

It
also makes sense to consider sets of rules where not every rule
is length-reducing, e.g., if one first needs to reshape a proof before a length-reducing rule can be applied, or if one translates between different proof calculi.
In this extended setting, termination is no longer trivially given, and thus one first needs to show that the introduced set of rules is terminating, which can for instance be achieved with the help of a reduction order~\cite{BaNi98,books/daglib/0008995}.
We show in this paper that the complexity of the small term reachability problem depends on which reduction order is used for this purpose.
More precisely, in \Cref{po:sect} we consider term rewriting systems that can be proved terminating using a polynomial order~\cite{Lankford79},
and show that in this case the small term reachability problem is in general N2ExpTime-complete, both
in the general and the confluent case. To prove the complexity upper bound, we actually need to restrict the employed polynomial orders slightly.
If the definition of the polynomial order uses only linear polynomials, then the complexity of the problem is reduced to NExpTime, where again
hardness already holds for confluent systems.
Here, as usual,
NExpTime (N2ExpTime) is the class of all decision problems solvable by a nondeterministic
Turing machine in $O(2^{p(n)})$ ($O(2^{2^{p(n)}})$) steps, where $n$ is the size of the
    problem and $p(n)$ is a polynomial in $n$.

\paragraph*{3.\ Small term reachability for  TRSs whose termination is shown
  by KBO}

In \Cref{kbo:sect}, we investigate the impact that using a Knuth-Bendix order (KBO)~\cite{KnuthBendix} for the termination proof has on the complexity of the small term 
reachability problem. In the restricted setting without unary function symbols of weight zero, the problem is in general PSpace-complete, 
again both in the general and the confluent case.
The complexity class PSpace consists of all decision problems solvable by a 
Turing machine in $O(p(n))$ space, where $n$ is the size of the problem and $p(n)$ is a polynomial in $n$.

\paragraph*{4.\ Large term reachability}

In order to investigate how much our complexity results depend on the exact formulation of the condition on the reachable terms, we consider
a variant of the problem, which we call the \emph{large term reachability problem}: given a term rewriting system $R$, a term $s$, and
a natural number $n$, decide whether there is a term $t$ of size $\geq n$ such that $s\too_R t$.
For length-reducing TRSs, this modification of the problem definition reduces the
complexity to (deterministic) linear time.
For the KBO, we obtain the same complexity results as for the small term reachability problem. For TRSs shown terminating with a linear polynomial
order, the complexity (NExpTime) stays the same if $n$ is assumed to be given in binary
representation, whereas the complexity goes down to PSpace for unary encoding of
$n$.
Similarly, for TRSs whose termination is proved by arbitrary polynomial orders, the
complexity goes down to ExpSpace (PSpace) for binary (unary) encoding of numbers. 
Actually, these upper bounds (i.e., ExpSpace for binary and PSpace for unary encoding)
do not depend on the use of polynomial orders, but hold for all terminating TRSs.

\bigskip

\paragraph*{Related work.}

In the area of term rewriting, 
\emph{reachability} usually refers to the following question (see, e.g.,
\cite{FeuilladeGT04,SternagelY19} and
\Cref{harmless} (2)):
given a TRS $R$ and two 
terms $s,t$, does
$s \too_R t$ hold? If $s$ and $t$ contain variables, then a related question is
\emph{feasibility} (see, e.g., \cite{LucasG18}), which asks
whether there exists a substitution $\sigma$ of the variables such that the
corresponding instance of $s$
rewrites to the corresponding instance of $t$ (i.e., such that $\sigma(s) \too_R \sigma(t)$ holds).
We study a related but different problem, since
we do not consider instantiations of the start term $s$ and we are interested in whether
\emph{some} term $t$ that is ``small enough'' can be reached, i.e., $t$ is not a fixed term given by the input. 
In general, \emph{reachability} is studied in many areas of Computer Science, with a whole
conference series devoted to the topic (see, e.g., \cite{RP2024}). However, we are not aware of any previous
work on the small term reachability problem for term rewriting.

The proofs of our results strongly depend on work on the derivational complexity of term rewriting systems, which links the
reduction order employed for the termination proof with the maximal length of reduction sequences as a function of the size of the start term
(see, e.g.,~\cite{DBLP:journals/tcs/Hofbauer92,DBLP:journals/iandc/Hofbauer03,DBLP:conf/rta/HofbauerL89,DBLP:journals/tcs/Lepper01}).
To obtain reasonable complexity classes, we restricted ourselves to reduction orders where the resulting bound on the
derivational complexity is not ``too high''.
In particular, we decided to start our investigation with the 
results of \cite{Geupel88,Lautemann88,DBLP:conf/rta/HofbauerL89}
on classical reduction orders, who show that
termination proofs with a (linear) polynomial order yield a double-exponential (exponential) upper bound on the length of derivation sequences whereas
termination proofs with a KBO without unary function symbols of weight zero yield an exponential such bound. 
Note that these results are proved in~\cite{DBLP:conf/rta/HofbauerL89} under the
assumption that the TRS is fixed.
We also make use of the term rewriting systems employed in the proofs showing that these bounds are tight. 
A connection between the derivational complexity of term rewriting systems and complexity classes
has been established in~\cite{DBLP:journals/jfp/BonfanteCMT01} for polynomial orders,
in~\cite{BonfanteMM05} for quasi-interpretations, 
and in~\cite{DBLP:conf/lpar/BonfanteM10} for Knuth-Bendix orders.
While this work considers a different problem since it views term rewriting systems as
devices for computing functions by generating a normal form, and uses them to characterize
complexity classes, 
the constructions utilized in the proofs in~\cite{DBLP:journals/jfp/BonfanteCMT01,DBLP:conf/lpar/BonfanteM10} are similar to the ones we use 
in our hardness proofs. A
notable difference between the two problems is that we are not specifically interested
in normal forms (i.e., irreducible terms), but in the question whether one can reach
``small'' terms $t$ (which need not be in normal form) from a given term $s$. 
Note that a small term $t$ might be reachable from $s$, though all normal forms of $s$ are not small. 
For this reason, the impact that confluence has on the obtained complexity class also differs
for the two problems:
while in our setting confluence only reduces the complexity in the case of length-reducing systems, in~\cite{DBLP:journals/jfp/BonfanteCMT01}
it also reduces the complexity (from the nondeterministic to the respective deterministic class) for the case of systems shown terminating with
a (linear) polynomial order.

\paragraph*{Comparison with previous conference publication.}

This article is based on a paper published at FSCD 2024 \cite{FSCD2024}, but extends and improves on the conference version in several respects.
While the conference paper restricts the attention to the small term reachability problem, the present paper also considers variants of
this problem, and in  particular the large term reachability problem. Our results demonstrate that the exact formulation of the condition
imposed on the reachable term may have a considerable impact on the complexity of the problem, although testing the condition for a given term
has the same complexity in both cases. In particular, it turns out that the complexity of the large term reachability problem may depend
on how the number that yields the size bound is encoded. In~\cite{FSCD2024}, the encoding
of numbers was not taken into account, since it has no impact on the complexity of the small term reachability
problem. In the present article, we prove this result and make it
explicit in the statement of our theorems.

Another improvement over our conference paper is that we take more care when stating and proving the upper complexity bounds. Note that
the known bounds on the maximal length of reduction sequences for TRSs shown terminating with a KBO or polynomial order are proved in~\cite{DBLP:conf/rta/HofbauerL89}
under the assumption that the TRS is fixed. For this reason, we also consider the small (large) term reachability problem for a fixed TRS. This is now
made clearer when showing the upper complexity bounds and stating our complexity
results. Here, it turns out that showing the upper complexity bounds
for the small term reachability problem in the case of polynomial orders is considerably more involved than claimed in~\cite{FSCD2024}, and requires an additional
(though not very restrictive) condition on the employed polynomial order (i.e., that
constants must be mapped to numbers $\geq 2$). In particular, we prove that using a (linear) polynomial order for showing termination of a TRS
does not only impose a double-exponential (exponential) upper bound on the length of derivation sequences, but also a double-exponential (exponential) upper bound on
the sizes of the terms that can be reached.
Though there has been work on deriving size bounds on reachable terms from the employed termination method
(see, e.g., \cite{BonfanteMM05,DBLP:conf/lpar/BonfanteM10}), to the best of our knowledge, this size bound has not been established before (note that
we do not impose any other constraints on the polynomial orders and we
also allow polynomials
with ``zero coefficients'' as long as the polynomials depend on all indeterminates).

\paragraph*{Overview of the paper}

In the next section, we briefly recall basic notions from term rewriting, including the definitions of polynomial and Knuth-Bendix orders.
In \Cref{problem:sect}, we introduce the small term reachability problem and show that it
is undecidable in general, but decidable for
terminating systems. Sections~\ref{length:red:sect}, \ref{po:sect}, and \ref{kbo:sect} respectively consider the length-reducing, polynomial order,
and Knuth-Bendix order case.
Finally, in \Cref{variants:sect} we consider variants of the small term reachability
problem in order to determine how our complexity results depend on the exact condition
in the considered reachability problem.
We conclude with a brief discussion of possible future work.

\section{Preliminaries}\label{prelim:sect}

We assume that the reader is familiar with basic notions and results regarding term rewriting.
In this section, we briefly recall the relevant notions, but refer the reader to \cite{BaNi98,books/daglib/0008995} for details.

Given
a finite set of \emph{function symbols} with associated \emph{arities} (called the \emph{signature})
and a disjoint set of \emph{variables}, terms are built in the usual way.
Function symbols of arity $0$ are also called \emph{constant symbols}. For example, if $x$
is a variable, $c$ is a constant symbol, and $f$ a binary
function symbol, then $c, f(x,c), f(f(x,c),c)$ are terms. The \emph{size} $|t|$ of a term $t$ is the number of occurrences of function symbols and variables
in $t$ (e.g., $|f(f(x,c),c)| = 5$). If $f$ is a function symbol or variable, then $|t|_f$ counts the number of occurrences of $f$ in $t$ (e.g., $|f(f(x,c),c)|_f = 2$).
As usual, nested applications of unary function symbols are often written as words. For example,
$g(g(h(h(g(x)))))$ is written as $gghhg(x)$ or $g^2h^2g(x)$.

A \emph{rewrite rule} (or simply rule) is of the form $l\to r$ where $l, r$ are terms such that $l$ is not a variable and 
every variable occurring in $r$ also occurs in $l$. In this paper, a term rewriting system (TRS) is a
\emph{finite} set of rewrite rules, and thus we do not mention finiteness explicitly when formulating our complexity results.
A given TRS $R$ induces the binary relation $\to_R$ on terms. We have $s \to_R t$ if
there is a rule $l \to r$ in $R$ such that $s$ contains a substitution instance $\sigma(l)$ of $l$ as
subterm, and $t$ is obtained from $s$ by replacing this subterm with $\sigma(r)$.
Recall that a \emph{substitution} is a mapping from variables to terms, which is homomorphically extended to a mapping from terms to terms.
For example, if $R$ contains the rule $hh(x) \to g(x)$, then $f(hhh(c),c)\to_R f(gh(c),c)$ and $f(hhh(c),c)\to_R f(hg(c),c)$.
The reflexive and transitive closure of $\to_R$ is denoted as $\too_R$, i.e., $s\too_R t$ holds if there are $n\geq 1$ terms
$t_1,\ldots,t_n$ such that $s=t_1, t=t_n$, and $t_i\to_R t_{i+1}$ for $i = 1,\ldots,n-1$.

Two terms $s_1, s_2$ are \emph{joinable} with $R$ if there is a term $t$ such that $s_i\too_R t$ holds for $i=1,2$.
The relation $\to_R$ is \emph{confluent} if $s\too_R s_i$ for $i=1,2$ implies that $s_1$ and $s_2$ are joinable with~$R$.
It is \emph{terminating} if there is no infinite reduction sequence $t_0\to_R t_1\to_R t_2\to_R\ldots$. If $\to_R$ is confluent (terminating), then we also
call $R$ confluent (terminating). The term $t$ is \emph{irreducible} if there is no term $t'$ such that $t\to_R t'$. If $s\too_R t$ and $t$ is irreducible,
then we call $t$ a \emph{normal form} of $s$. If $R$ is confluent and terminating, then every term has a unique normal form. If $R$ is
terminating, then its confluence is decidable~\cite{KnuthBendix}. Termination can be proved using a \emph{reduction order}, which is a well-founded order $\succ$ on terms such that
$l \succ r$ for all $l\to r\in R$ implies $s \succ t$ for all terms $s,t$ with $s\to_R t$. Since $\succ$ is well-founded, this then implies termination of $R$.
If $l \succ r$ holds for all $l\to r\in R$, then we say that $R$ can be \emph{shown terminating} with the reduction order $\succ$.
The following is a simple reduction order.

\begin{exa}\label{length:red:ord}
If we define $s \succ t$ if $|s| > |t|$ and $|s|_x \geq |t|_x$ for all variables $x$, then~$\succ$ is a reduction order (see Exercise~5.5 in~\cite{BaNi98}). 
For example, $hh(x) \succ g(x)$, and thus the TRS
$R = \{hh(x) \to g(x)\}$ is terminating. As illustrated in Example~5.2.2 in~\cite{BaNi98}, the condition on variables is needed to
obtain a reduction order. 
\end{exa}
This order can only show termination of \emph{length-reducing} TRSs $R$, i.e.,
where $s\to_R t$ implies $|s| > |t|$. We now recapitulate the definitions of more powerful reduction orders~\cite{BaNi98,books/daglib/0008995}.

\paragraph*{Polynomial orders}
To define a polynomial order, one assigns to every $n$-ary function symbol $f$ a polynomial $P_f$ with coefficients
in the natural numbers $\N$ and $n$ indeterminates $x_1,\ldots,x_n$ such that $P_f$ depends on all these
indeterminates. To ease readability, we usually write $x$ instead of $x_1$ if $n =1$, and
we 
write $x, y$ instead of $x_1, x_2$ if $n = 2$.
To ensure that dependence on all indeterminates implies (strong) monotonicity of the polynomial order, we require that constant symbols $c$ must
be assigned a polynomial of degree $0$ whose coefficient is $> 0$. Such an assignment also yields an assignment of polynomials $P_t$ to terms $t$.

\begin{exa}\label{Ex:PolynomialOrders}
  Assume that $+$ is binary, $s,d,q$ are unary, and $0$ is a constant.
  We assign the polynomial 
$P_+ = x+2y+1$ to $+$, $P_s = x+3$ to $s$, $P_d = 3x+1$ to $d$,  $P_q = 3x^2+3x+1$ to $q$,
and
$P_0 = 4$ to $0$.
For the terms $l = q(s(x))$ and $r = q(x)+s(d(x))$ we then obtain the associated polynomials
$P_l = 3(x+3)^2 +3(x+3)+1 = 3x^2 +21x+37$ and
$P_r = 3x^2 +3x+1+2(3x+1+3)+1=3x^2 +9x+10$.
\end{exa}

The polynomial order induced by such an assignment is defined as follows: $t\succ t'$ if $P_t$ evaluates to a larger
natural number than $P_{t'}$ for every assignment of natural numbers $>0$ to the indeterminates of $P_t$ and $P_{t'}$. In our example,
the evaluation of $P_l$ is obviously always larger than the evaluation of $P_r$, and thus $l\succ r$.
Polynomial orders are reduction orders, and thus can be used to prove termination of TRSs (see, e.g., Section~5.3 of~\cite{BaNi98}).

As mentioned above, for our complexity upper bounds, we need the polynomial order that is
used to show termination of the TRS to satisfy an additional restriction.

\begin{defi}\label{scomp:def}
We say that a polynomial order is \emph{\scomp} if for every constant symbol the
assigned polynomial of degree $0$ has a coefficient $\geq 2$ (i.e., every constant symbol must be
assigned a number $\geq 2$).
\end{defi}

This restriction provides us with a double-exponential bound (in the size of $s$) on the sizes of the terms reachable with $\too_R$ from $s$
(see \Cref{size:bound:lem} in \Cref{po:sect}).
The polynomial order of \Cref{Ex:PolynomialOrders} is \scomp, since $P_0 = 4 \geq 2$.

\paragraph*{Knuth-Bendix orders}

To define a Knuth-Bendix order (KBO), one must assign a weight $w(f)$ to all function
symbols and variables $f$, and define a strict order $>$
on the function symbols (called \emph{precedence}) such that the following is satisfied:
 
\begin{itemize}
\item
  All weights $w(f)$ are non-negative real numbers, and there is a weight $w_0>0$ such that $w(x) = w_0$ for all variables $x$ and
  $w(c) \geq w_0$ for all constant symbols $c$.
\item
  If there is a unary function symbol $h$ with $w(h) = 0$, then $h$ is the greatest element w.r.t.\ $>$, i.e., $h > f$ for
  all function symbols $f\neq h$. Such a unary function symbol $h$ is then called a \emph{special} symbol. Obviously, there
  can be at most one special symbol.
\end{itemize}
Since in this paper we only consider KBOs without special symbol, we restrict our definition of KBOs to this case.
For any weight function $w$, we define its extension to terms as
$w(u) := \sum_{f\, \text{occurs in}\, u}w(f) \cdot |u|_f$  for all terms $u$.
Then a given weight function $w$ and strict order~$>$ without special symbol induces the
following KBO $\succ$:
\smallskip

\noindent
$s\succ t$ if 
$|s|_x \geq |t|_x$ for all variables $x$ and
\begin{itemize}
\item $w(s) > w(t)$, or
\item $w(s) = w(t)$ and one of the following two conditions is satisfied:
\begin{itemize}
\item $s = f(s_1,\ldots,s_m)$, $t = g(t_1,\ldots,t_n)$, and $f>g$.
\item $s = f(s_1,\ldots,s_m)$, $t = f(t_1,\ldots,t_m)$, and there is $i, 1\leq i\leq m,$ such that\\
      $s_1=t_1,\ldots, s_{i-1} = t_{i-1}$, and $s_i\succ t_i$.
\end{itemize}
\end{itemize}
A proof of the fact that KBOs are reduction orders can, e.g., be found in Section~5.4.4 of ~\cite{BaNi98}.

\begin{exa}\label{down:cont:ex}
Let $0,1,1'$ be unary function symbols and $\varepsilon$ a constant symbol, and consider the following TRS, which is similar to the one
introduced in the proof of Lemma~7 in~\cite{DBLP:conf/lpar/BonfanteM10}:
\[
R = \{
1(\varepsilon) \to 0(\varepsilon),\ \
0(\varepsilon) \to 1'(\varepsilon),\ \
0(1'(x)) \to 1'(1(x)),\ \
1(1'(x)) \to 0(1(x)) \}.
\] 
Basically, this TRS realizes a binary down counter, and thus it is easy to see that, starting with the binary representation $10^n(\varepsilon)$ of the number $2^n$,
the TRS $R$ can make $\geq 2^n$ reduction steps to arrive at the term $0^{n+1}(\varepsilon)$. For example,
$
100(\varepsilon) \to_R 101'(\varepsilon) \to_R 11'1(\varepsilon) \to_R 011(\varepsilon)
\to_R 010(\varepsilon) \to_R 011'(\varepsilon) \to_R 001(\varepsilon) \to_R 000(\varepsilon).
$
Termination of $R$ can be shown using the following
KBO: assign weight $1$ to all function symbols and variables, and use the precedence order $1 > 0 > 1'$.
\end{exa}

\section{Problem definition and (un)decidability results}\label{problem:sect}

In this paper, we mainly investigate the complexity of the following decision problem. 

\begin{defi}
Let $R$ be a TRS.
Given 
a term $s$ and a natural number $n$,
the \emph{small term reachability problem} for $R$ asks whether there exists a term $t$
such that $s\too_R t$ and $|t|\leq n$.
\end{defi}
The name ``small term reachability problem'' is motivated by the fact that we want to use the TRS $R$
to turn a given term $s$ into a term whose size is as small as possible. The introduced problem is the decision variant
of this computation problem. A solution to the computation problem, which computes a term $t$ of minimal size reachable with $R$
from $s$, of course also solves the decision variant of the problem. Thus, complexity lower bounds for the decision problem
transfer to the computation problem. 

It is easy to see that this problem is in general undecidable, but decidable for terminating TRSs.
For non-terminating systems, confluence is not sufficient to obtain decidability.

\begin{prop}\label{decidability-prop}
The small term reachability problem is in general undecidable for confluent TRSs, but is
decidable for systems that are terminating. 
\end{prop}

\begin{proof}
Undecidability in the general case follows, e.g., from the fact that TRSs can simulate Turing
machines (TMs) \cite{HuetLankford1978}. (We will also use Turing machines for the proofs
of the hardness results in the remainder of the paper.)
More precisely, the reduction introduced in Section~5.1.1 of~\cite{BaNi98} transforms a given TM \M into a TRS 
$R_\M$ such that (among other things) the following holds: there is an infinite run of \M on the empty input iff there is an infinite
reduction sequence of $R_\M$ starting with the term $s_0$ that encodes the initial configuration of \M for the empty input. In addition,
if \M is deterministic, then $R_\M$ is confluent. We can now add rules to $R_\M$ that apply to all terms encoding a halting configuration of 
\M, and trigger further rules that reduce such a term to one of size $1$. Since the term $s_0$ has size larger than one and the rules
of $R_\M$ never decrease the size of a term, this yields a reduction of the (undecidable) halting problem for deterministic TMs
to the small term reachability problem for confluent TRSs. If we apply this reduction to a TM for which the halting problem is undecidable
(e.g., a universal TM), then this reduction shows that there are fixed TRSs $R$ for which the small term reachability problem is undecidable.

Given a terminating TRS $R$ and a term $s$, we can systematically generate
all terms reachable from $s$ by iteratively applying $\to_R$. Since $R$ is finite, $\to_R$ is finitely branching. Together
with termination, this means (by König's Lemma) that there are only finitely many terms reachable with $R$ from $s$
(see Lemma~2.2.4 in~\cite{BaNi98}). We can then check whether, among them, there is a term of size at most~$n$.
\end{proof}

In the following three sections, we study the \emph{complexity} of the small term reachability problem for terminating
TRSs, depending on how their termination can be shown. In general, if one is interested in the complexity of a problem 
whose formulation involves a number $n$, this complexity may depend on whether this number
is represented in unary or binary encoding.
In the first case, the contribution of the number to the size of the input is $n$, whereas it is $\log_2 n$ in the second case.
For this reason, there can potentially be an exponential difference between the complexity
(measured as a function of the size of the input)
of such a decision problem depending on the assumed representation of the number.
For the instances of the small term reachability problem considered in the next three sections, it turns out that the
encoding of the number $n$ does not have an impact on the complexity. However, in \Cref{variants:sect} we consider
a variant of the problem where using unary encoding leads to a decrease of the complexity.

\section{Length-reducing term rewriting systems}\label{length:red:sect}

In this section, we investigate the complexity of the small term reachability problem
for length-reducing TRSs, i.e., TRSs where each rewrite step decreases the size of the
term.
 
We start with showing an \emph{NP upper bound}.
Let $R$ be a length-reducing TRS and let
$s, n$ be an instance of the small term reachability problem for $R$.
Since $R$ is length-reducing,
the length $k$ of any rewrite sequence
$s = s_0 \to_R s_1 \to_R s_2 \to_R \ldots \to_R s_k$ issuing from $s$ is bounded by $|s|$
and we have $|s_i| < |s|$ for all $1\leq i \leq k$.
Thus, the following yields an NP procedure for deciding the
small term reachability problem:
\begin{itemize}
\item
  guess $k$ terms $s_1, \ldots, s_k$ with $k \leq |s|$ and $|s_i| < |s|$ for all $1\leq i
  \leq k$;
\item check whether  $s = s_0 \to_R s_1 \to_R s_2 \to_R \ldots \to_R s_k$ holds and
  whether $|s_k|\leq n$. If the answer is ``yes'' then accept, and reject otherwise.
\end{itemize}

\begin{lem}\label{lenred:upper:bound}
The small term reachability problem is in NP for length-reducing TRSs, both for unary and binary encoding of numbers.
\end{lem}

\begin{proof}
  First, note that the
terms $s_1, \ldots, s_k$ can be generated by
an NP-procedure since $k$ and the sizes of the occurring terms
are bounded by the size $|s|$ of the input term $s$.
Checking whether
$s_{i-1} \to_R s_{i}$ holds for 
 $1\leq i
  \leq k$ can also be done in polynomial time:
  One has to check whether there is a
   rewrite rule in $R$, a position in $s_{i-1}$,  a substitution such
   that the rule is applicable at this position, and whether its application yields $s_i$. 
   Since the sizes of $s_{i-1}$ and $s_i$
   are bounded by the size of $s$, this can clearly be achieved in polynomial time.

Moreover, the test whether $|s_k|\leq n$ holds can be realized in linear time in the size of $s_k$ and the size of the representation of $n$,
both for unary and binary representation of $n$. Basically, one can traverse a textual representation of $s_k$ from left to right, and for
every encountered function symbol or variable subtract $1$ from the number, starting with $n$. If $0$ is
obtained before the end of the term has been reached,
then the test answers ``no,'' and otherwise it answers ``yes.'' Subtraction of $1$ from a number is clearly possible in time linear
in the size of the representation of the number, both for unary and binary encoding. Since the size of $s_k$ is bounded by the size of $s$, this shows
that the test $|s_k|\leq n$ can be performed in linear time in the size of $s$ and the representation of $n$.
\end{proof}

If the length-reducing system $R$ is confluent, then it is sufficient to generate an
arbitrary \emph{terminating} (i.e., maximal) rewrite sequence 
starting in $s$, i.e., a sequence $s = s_0\to_R s_1 \to_R s_2 \to_R \ldots \to_R s_k$ such that $s_k$ is irreducible. Obviously,
we have $k\leq |s|$, and thus such a sequence can be generated in polynomial time. We claim that there
is a term $t$ of size $\leq n$ reachable from $s$ iff $|s_k|\leq n$. Otherwise, the smallest term $t$ reachable from $s$
is different from $s_k$. But then $t$ and $s_k$ are both reachable from $s$, and thus, they must be joinable due to the
confluence of $R$.
As $s_k$ is irreducible, this implies $t \to_R^* s_k$ and thus, $|t| \geq |s_k|$, i.e.,
$t$ is not smaller than $s_k$.

\begin{prop}\label{confluence:prop}
For confluent length-reducing TRSs, the small term reachability problem can be decided
in deterministic polynomial time, both for unary and binary encoding of numbers.
\end{prop}

In general, however, there are (non-confluent) length-reducing TRSs for which the problem is \emph{NP-hard}.
We prove this by showing that any polynomially time bounded nondeterministic Turing machine 
can be simulated by a length-reducing TRS. 
Thus, assume that $\M$ is such a TM and that its time-bound is given by the polynomial $p$.
As in~\cite{BaNi98} we assume that in every step \M either moves to the left or to the
right, where the tape of the TM is infinite in both directions. 
In addition, we assume without loss of generality that \M has exactly one accepting state \acc.
We view the tape symbols of \M as unary function symbols and the states of \M as binary function symbols. 
We assume that $q_0$ is the initial state of \M and that $b$ is the blank symbol.
Furthermore, let $\#$ be a constant symbol and $f$ be a unary function symbol different from the tape symbols.

Configurations of the TM \M are represented by terms of the form
\[ b^{n}  a_1'\ldots a_{\ell'}' ( q(a_1 \ldots a_{\ell}b^{m}(\#), f^k(\#))).\]
This term represents the TM \M in the state $q$. The
first argument of a state symbol like $q$ corresponds to the part of the tape that starts at the
position of the head. Thus, in the term above, $a_1$ is the tape symbol at the position of the
head
and $a_2 \ldots a_\ell$ are the symbols to the right of it.
The symbols $a_1'\ldots a_{\ell'}'$ are the symbols to the left of the position of the
head. Moreover, we have $n$ blank symbols on the tape to the left of $a_1'$ and $m$ blank
symbols to the right of $a_{\ell}$, where $n$ and $m$ are chosen large enough such that
the TM does not reach the end of the represented tape. The second argument
$f^k(\#)$
of a state
symbol like $q$ is a unary down counter from which one $f$ is removed in every step that
\M makes. 
This is needed to ensure that
the constructed TRS is length-reducing.
So this 
counter indicates how many steps are still possible with the
TM \M  (i.e., for the term above, $k$ further steps
are possible).

Given an input word $w = a_1\ldots a_\ell$ for \M, we now construct the term
\[
t(w) := b^{p(\ell)}(q_0(a_1\ldots a_\ell b^{p(\ell)-\ell}(\#),f^{p(\ell)}(\#))).
\]
Intuitively, the starting $b$ symbols together with the first argument of $q_0$ in $t(w)$
provide a tape that is large enough for a $p(\ell)$-time bounded TM 
to run on for the given input $w$ of length $\ell$.
The counter $f^{p(\ell)}(\#)$ is large enough to allow \M to make the maximally possible
number of $p(\ell)$ steps.

Basically, we now express the transitions of \M as usual by rewrite rules (as, e.g., done in Definition~5.1.3 of~\cite{BaNi98}), but with
three differences:
\begin{itemize}
\item
  since the term $t(w)$ provides enough tape for a TM that can make at most $p(\ell)$ steps, the special cases that
  treat a situation where the end of the represented tape is reached and one has to add a blank are not needed;
\item
  since we fix as start term $t(w)$ a configuration term (i.e., a term that encodes a configuration of the TM), the additional effort expended in~\cite{BaNi98}
  to deal with non-configuration terms (by using copies of symbols with arrows to the left or right) is not needed;
\item
  we have the additional counter in the second argument, which removes one $f$ in every step, and thus ensures that
  rule application is length-reducing.
\end{itemize}
The TRS $\RM$ that simulates \M has the following rewriting rules:
\begin{itemize}
\item
  For each transition $(q,a,q',a',r)$ of \M it has the rule\ \ $q(a(x),f(y))\rightarrow
  a'(q'(x,y))$. Thus, the tape symbol $a$ is replaced by $a'$ and the head of the TM is
  now at the position to the right of it.
\item
  For each transition $(q,a,q',a',l)$ of \M it has the rule\ \ $c(q(a(x),f(y)))\rightarrow q'(ca'(x),y)$
  for every tape symbol $c$ of \M. Thus, $a$ is replaced by $a'$ and the head of the TM is
  now at the position to the left of it.
\end{itemize}
Note that the blank symbol $b$ is also considered as a tape symbol of $\M$.

In addition, we add rules to $\RM$ that can be used to generate the term $\#$, which has size~1, whenever \acc is reached:
\begin{itemize}
\item
  $a(\acc(x,y)) \to \acc(x,y)$ for every tape symbol $a$ of \M,
\item
  $\acc(x,y) \to \#$.
\end{itemize}
The following is now easy to see.

\begin{lem}\label{hardness:tech:lem}
The term $t(w)$ can be rewritten with $\RM$ to a term of size $1$ iff \M accepts the word $w$.
\end{lem}

\begin{proof}
It is easy to see that $\RM$ simulates \M in the sense that there is a run of \M on input $w = a_1\ldots a_\ell$ that reaches the accepting state
\acc iff there is a rewrite sequence of $\RM$ starting with $t(w)$ that reaches a term of the form $u(\acc(t,t'))$, where $u$ is a word
over the tape symbols of \M and $t, t'$ are terms. Note that the assumption that \M is $p(\ell)$-time bounded together with the construction
of $t(w)$ ensures that there is enough tape space and the counter is large enough for the simulation of \M to run through
completely.

Thus, if \M accepts $w=a_1\ldots a_\ell$, then we can rewrite $t(w)$ with $\RM$ into a term of the form $u(\acc(t,t'))$, and this term can
then be further rewritten into $\#$, which has size $1$. 
If \M does not accept $w= a_1\ldots a_\ell$, then the state \acc cannot be reached by any run of \M starting with this word. 
Thus, all terms reachable from $t(w)$ with the rules of $\RM$ that simulate \M are of the form $u(q(t,t'))$ for states $q$ different from \acc.
The rules of $\RM$ of the second kind are thus not applicable, and the terms of the form $u(q(t,t'))$ clearly have size $> 1$.
\end{proof}

We are now ready to show the corresponding complexity lower bound.

\begin{lem}\label{hardness:main:lem}
  There are length-reducing TRSs
  for which the
  small term reachability problem is NP-hard, both for unary and binary encoding of numbers.
\end{lem}

\begin{proof}
  Let $\Pi$ be an NP-hard problem.
  We show that there is a length-reducing TRS $R$ such that $\Pi$ can be reduced in polynomial time to a small term reachability problem for $R$.
Let \M be the nondeterministic Turing machine that is an NP decision procedure for $\Pi$, and let $p$ be
the polynomial that bounds the length of runs of \M. We can construct the length-reducing TRS \RM
as described above. Given a word $w=a_1\ldots a_\ell$, we can compute the term $t(w)$ in polynomial time, and
\Cref{hardness:tech:lem} implies that this yields a reduction function from $\Pi$ to the
small term reachability problem for the length-reducing TRS \RM where we use the number $n = 1$, whose representation is of constant
size both for unary and binary encoding of numbers.
\end{proof}

Combining the obtained upper and lower bounds, we thus have determined the exact complexity of the problem under consideration.

\begin{thm}\label{completeness:main:thm}
The small term reachability problem is in NP for length-reducing TRSs, and there are such
TRSs for which the small term reachability problem 
is NP-complete. These results hold both for unary and binary encoding of numbers.
\end{thm}

To show that a given TRS $R$ is length-reducing, one can, for example, use the reduction order of \Cref{length:red:ord}.
This order also applies to the TRS \RM introduced above.

\section{Term rewriting systems shown terminating with a polynomial order}\label{po:sect}

An interesting question is whether similar results can be obtained for TRSs whose termination
can be shown using a reduction order from a class of such orders that provides an upper bound on the length of reduction sequences.
For example, it is known that a proof of termination using a polynomial order yields a double-exponential upper bound on the
length of reduction sequences \cite{Geupel88,Lautemann88,DBLP:conf/rta/HofbauerL89}. One possible conjecture could now be that, 
for TRSs whose termination can be shown using a polynomial order, the small term
reachability problem is in general N2ExpTime-complete.

The \emph{upper bound} can in principle be established similarly to the case of length-reducing systems: 
again, one needs to guess a reduction sequence, but now of at most double-exponential length, and
then check the size of the obtained term. However, proving that this yields a nondeterministic double-exponential time procedure 
for solving the small term reachability problem 
is less trivial than in the case of length-reducing TRSs. The reason is that the sizes of the terms in this sequence need no longer be bounded
by the size of the start term $s$. Instead, 
we obtain a double-exponential bound on the sizes of these terms, but proving this needs quite some effort and requires the additional restriction
that the employed order is \scomp.
We start by establishing a double-exponential bound on the number of non-unary function symbols occurring in such terms.
Given a term $t$, its \emph{nu-size} $\nusz{t}$ counts the number of occurrences of non-unary function symbols and variables in $t$.

\begin{lem}\label{nu-size:bound}
  Let $\succ$ be a \scomp
  polynomial order, and $s, s'$ be terms such that $s\succ s'$.
Then the nu-size of $s'$ is double-exponentially bounded by the size of $s$. 
\end{lem}

\begin{proof}
  For every term $t$ and natural number $a \geq 1$, let $P_t$
again
  be the polynomial associated with $t$ by the polynomial interpretation that induces $\succ$,
and let $\pi_a(t)$ be the evaluation of $P_t$ with $a$ substituted for all indeterminates of $P_t$.
It is shown in~\cite{DBLP:conf/rta/HofbauerL89} that for any
$a \geq 1$,
$\pi_a(t)$ is double-exponentially bounded by the size of $t$ 
(see also the proof of Proposition~5.3.1 in~\cite{BaNi98}). 

We claim that, if $a\geq 2$, then $\pi_a(t) > \nusz{t}$ for all terms $t$. 
Consequently, $\pi_a(s) > \pi_a(s') > \nusz{s'}$ implies that the nu-size of $s'$ is
double-exponentially bounded by the size of $s$.
Note that the inequality $\pi_a(s) > \pi_a(s')$ holds since $s\succ s'$. 

To prove the claim, we
use structural induction on $t$. If $t$ is a variable then we have $\pi_a(t) = a \geq 2 > 1 = \nusz{t}$.
If $t$ is a constant, then $\pi_a(t) = P_t \geq 2 > 1 = \nusz{t}$, where $P_t \geq 2$ holds since $\succ$ is assumed to be \scomp.

In the induction step, we first consider the case $t = f(t_1,\ldots,t_n)$ for $n \geq 2$.
Since $P_f$ depends on all its indeterminates, it is a sum of monomials with coefficients $\geq 1$ such that each indeterminate occurs in at least one monomial.
Using the fact that $a_{i_1} \cdot \ldots \cdot a_{i_k}\geq a_{i_1} + \ldots + a_{i_k}$ holds for all $a_{i_1},\ldots,a_{i_k}\geq 2$,
we can deduce that $P_f(a_1,\ldots,a_n) \geq a_1+ \ldots + a_n$ for all $a_1,\ldots,a_n\geq 2$. 
Consequently, 
$\pi_a(f(t_1,\ldots,t_n)) = P_f(\pi_a(t_1), \ldots, \pi_a(t_n)) \geq
\pi_a(t_1) + \ldots + \pi_a(t_n)$.\footnote{%
Note that \scompa ensures that $\pi_a(t_1), \ldots, \pi_a(t_n)\geq 2$ is satisfied.
}
The induction hypothesis yields $\pi_a(t_i) > \nusz{t_i}$ for $i = 1,\ldots,n$,
and thus $\pi_a(t_1) + \ldots + \pi_a(t_n) \geq (\nusz{t_1} +1) + \ldots + (\nusz{t_n} +1) > (\nusz{t_1} + \ldots + \nusz{t_n}) +1 = \nusz{f(t_1,\ldots,t_n)}$.
The strict inequality holds since $n\geq 2$.

If $t = f(t_1)$ for a unary function symbol $f$, then $\pi_a(t) = \pi_a(f(t_1)) \geq 
\pi_a(t_1) > \nusz{t_1} = \nusz{t}$, where the strict inequality holds by the induction hypothesis.
\end{proof}

The following example illustrates that \scompa and the restriction to counting only
non-unary function symbols in $s'$ are needed for this lemma to hold.

\begin{exa}
First, we show that \Cref{nu-size:bound} does not hold if we replace the nu-size of $s'$ with the
size of $s'$.
In fact there is no bound function $b$ such that 
$s\succ s'$ implies $|s'| \in O(b(|s|))$
for all terms $s, s'$. The reason is that there exist terms $s$
for which we have $s \succ s'$ for infinitely many terms $s'$ that have
arbitrarily large sizes.
As an example, consider a \scomp polynomial order
that assigns the polynomial $P_{f} = x$ to the unary function symbol $f$, $P_c = 2$ to the
constant $c$, and $P_d = 3$ to the constant $d$.
Let $s$ be the term $d$ and as $s'$, we consider all terms of the form 
$s_k := f^k(c)$ for
every $k\geq 1$.
Clearly, we have $d \succ s_k$ for all $k \geq 1$ since $P_d = 3 > 2 = P_{s_k}$. However,
the size of $s_k$ is $|s_k| = k +1$, i.e., we have
 $d \succ s_k$ for infinitely many terms $s_k$ that have
arbitrarily large sizes.

Second, we show that the lemma does not hold without the \scompa assumption.
In fact, again then there is no bound function $b$ such that 
$s\succ s'$ implies $|s'|_{nu} \in O(b(|s|))$
for all terms $s, s'$.
Similar to the reasoning above, the reason is that then there exist terms $s$
for which we have $s \succ s'$ for infinitely many terms $s'$ that have
arbitrarily large nu-sizes.
As an example, 
assume that $f$ is a binary function symbol such that $P_f = x \cdot y$, $c$ is a
constant symbol with $P_c = 1$ (which violates \scompa), and
$d$ is a constant with associated polynomial $P_d = 2$.
Let $s$ again be the term $d$ and as $s'$, 
we consider the terms
$t_k$ for $k\geq 0$, which we define
by induction: $t_0 := c$ and $t_{k+1} := f(t_k,t_k)$.
Clearly, we have $d \succ t_k$ for all $k \geq 0$ since $P_d = 2 > 1 = P_{t_k}$.
However,
the nu-size of $t_k$ is $\nusz{t_k} = |t_k| = 2^{k+1}-1$,
i.e., we have
$d \succ t_k$ for infinitely many terms $t_k$ that have
arbitrarily large nu-sizes.
\end{exa}

However, as in the first part of this example, terms that are large since they contain many unary function symbols must also have a large nesting depth. 
Thus, if we additionally have a bound on this depth, then we also obtain a bound on the size.

To be more precise, for a term $t$ its \emph{depth} $\dep(t)$ is the maximal nesting of function symbols in $t$, i.e.,
$\dep(c) = \dep(x) = 0$ for constants $c$ and variables $x$, and $\dep(f(t_1,\ldots,t_n)) = 1 + \max_{1\leq i\leq n}\dep(t_i)$.
As shown in~\cite{DBLP:conf/cade/CichonL92}, applying a rewrite step can increase the depth of a term only by a constant that is determined 
by the given TRS.

\begin{lemC}[\cite{DBLP:conf/cade/CichonL92}]
For every TRS $R$ there is a constant $m_R$ such that $s\to_R t$ implies $\dep(t)\leq\dep(s)+m_R$.
\end{lemC}

Consequently, if we start with a term $s$ and apply at most double-exponentially many (in the size of $s$) rewrite steps, then the depth of the obtained term is at most
double-exponential in the size of $s$.

\begin{lem}\label{depth:bound}
Let $R$ be a TRS whose termination can be shown using a polynomial order. If $s, t$ are terms such that
$s\too_R t$, then the depth of $t$ is double-exponentially bounded by the size of $s$.
\end{lem}

Bounds on the depth and the nu-size yield the following bound on the size of a term.

\begin{lem}
If $\nusz{t} = m$ and $\dep(t) = n$, then $|t| \leq m + m \cdot n$.
\end{lem}

\begin{proof}
Consider the tree representation of $t$. For every leaf (labeled by a constant or variable), consider the number
of unary function symbols occurring on the unique path from the root to this leaf. This number is clearly bounded by $\dep(t)$ and the
number of leaves is bounded by $\nusz{t}$. Thus, the number obtained by summing up all these numbers is bounded by $m\cdot n$,
and this sum is an upper bound on the number of unary function symbols occurring in~$t$.
\end{proof}

Due to this result, \Cref{nu-size:bound,depth:bound} yield the desired  bound on the sizes of reachable terms.
Indeed, if both the nu-size of $t$ and the depth of $t$ are double-exponentially bounded by the size of $s$, then the size of $t$ is also
double-exponentially bounded by the size of $s$.

\begin{prop}\label{size:bound:lem}
Let $R$ be a TRS whose termination can be shown using a \scomp polynomial order. If $s, t$ are terms such that
$s\too_R t$, then the size of $t$ is double-exponentially bounded by the size of $s$.
\end{prop}

We are now ready to prove the upper complexity bound.

\begin{lem}
The small term reachability problem is in N2ExpTime for TRSs whose termination can be shown using a \scomp polynomial order $\succ$, both for unary and binary encoding of numbers.
\end{lem}

\begin{proof}
This proof is similar to the one of \Cref{lenred:upper:bound}, but now the length of the
guessed reduction sequence
$
s = s_0\to_R s_1 \to_R s_2 \to_R \ldots \to_R s_k
$
issuing from $s$ is of at most double-exponential length $k$ in $|s|$ and the sizes of the terms $s_i$ are double-exponentially bounded by the size of $s$
according to \Cref{size:bound:lem}. 
Thus, choosing a successor term of a term $s_i$ in the sequence can now be done by an N2ExpTime procedure, and the final
test $|s_k|\leq n$ requires double-exponential time, both for unary and binary encoding of numbers, using the algorithm sketched in the proof of \Cref{lenred:upper:bound},
but now applied to a term of at most double-exponential size. Overall, this yields an N2ExpTime procedure.
\end{proof}

Note that, in this lemma, we assume that the TRS $R$ is fixed. In fact, the double-exponential upper bound on the length of reduction sequences and the size 
of terms is of the form $2^{2^{c|s|}}$, where the number $c$ depends on parameters of the polynomial order used to show termination of $R$ rather than on the size of $R$ directly.
If $R$ is fixed, we can take a fixed polynomial order showing its termination, and thus $c$ is a constant. Otherwise, 
there is no clear correspondence between the value of $c$ and the size of $R$.

Regarding the \emph{lower bound}, the idea is to use basically the same approach as employed in \Cref{length:red:sect},
but generate a double-exponentially large tape and a double-exponentially large counter with the help of a TRS
whose termination can be shown using a polynomial order. For this, we want to re-use the original system introduced by Hofbauer and Lautemann
showing that the double-exponential upper bound is tight (see Example~5.3.12 in~\cite{BaNi98}). 

\begin{exa}\label{ho:la:2ex}
Let $\RHL$ be the TRS consisting of the following rules:
\[
\begin{array}{r@{\;}c@{\;}l@{\qquad}r@{\;}c@{\;}l@{\qquad}r@{\;}c@{\;}l}
  x+0 &\to& x, & d(0) &\to& 0,& q(0) &\to& 0,\\
 x+s(y) &\to& s(x+y), & 
 d(s(x)) &\to& s(s(d(x))), &
 q(s(x)) &\to& q(x) + s(d(x)).
\end{array}
\]
The TRS $\RHL$
intuitively defines the arithmetic functions addition ($+$), double ($d$), and square ($q$) on non-negative integers. 
Thus, it is easy to see that the term $t_n := q^n(s^2(0))$ can be reduced to $s^{2^{2^n}}(0)$.
The polynomial order in \Cref{Ex:PolynomialOrders} shows
termination of $\RHL$.
\end{exa}

Now, assume that $\M$ is a double-exponentially time bounded nondeterministic TM and that its time-bound is $2^{2^{p(\ell)}}$ 
for a polynomial $p$, where $\ell$ is the length of the input word. 
Given an input word $w = a_1\ldots a_\ell$ for \M, we construct the term
\[
t(w) := q_1^{p(\ell)}(bb(q_0(a_1\ldots a_\ell q_2^{p(\ell)}(bb(\#)),q_3^{p(\ell)}(ff(\#))))).
\]
The idea underlying this definition is that the term $q_1^{p(\ell)}(bb(q_0(\cdot))$ can be used to generate a tape segment
before the read-write head of the TM (marked by the state $q_0$) with $2^{2^{p(\ell)}}$ blanks using the following modified
version $R_1$ of $\RHL$:
\[
\mbox{\footnotesize $\begin{array}{@{\!}l@{}r@{\;}c@{\;}l@{}r@{\;}c@{\;}l@{\;\;}r@{\;}c@{\;}l@{\!\!\!}r@{\;}c@{\;}l}
\{&q_0(y_1,y_2)+_1 q_0(z_1,z_2) &\to& q_0(y_1,y_2), &
d_1(q_0(z_1,z_2)) &\to& q_0(z_1,z_2),&
q_1(q_0(z_1,z_2))&\to& q_0(z_1,z_2),\\
& b(x)+_1 q_0(z_1,z_2) &\to&
b(x), &
 d_1(b(x)) &\to& b(b(d_1(x))),&
 q_1(b(x)) &\to& q_1(x) +_1  b(d_1(x)),\\
 & x+_1 b(y) &\to& b(x+_1 y) \;
 \}.
\end{array}$}
\]
Here $b$ plays the rôle of the successor function $s$ in \RHL, terms of the form $q_0(\cdot)$ play the rôle of the zero $0$ in \RHL,
and $+_1$, $d_1$, and $q_1$ correspond to addition, double, and square. 
Instead of the rule $x +_1 q_0(z_1,z_2) \to x$ we considered two rules for the
case where $x$ is built with $q_0$ or with $b$, respectively. The reason will become clear
later when we consider the restriction to confluent TRSs.
\Cref{R1 Lemma} is an easy consequence of our observations regarding \RHL.

\begin{lem}\label{R1 Lemma}
For any two terms $t_1,t_2$, we can rewrite the term $q_1^{p(\ell)}(bb(q_0(t_1,t_2)))$ with $R_1$ into the term
$b^{2^{2^{p(\ell)}}}(q_0(t_1,t_2))$.
\end{lem}

Next, we define a copy of \RHL that allows us to create a tape segment with $2^{2^{p(n)}}$ blanks to the right of the input word:
\[
\begin{array}{l@{}r@{\;}c@{\;}l@{\quad}r@{\;}c@{\;}l@{\quad}r@{\;}c@{\;}l}
R_2 :=
\{&\# +_2 \# &\to& \#,&
d_2(\#) &\to& \#,&
 q_2(\#) &\to& \#, \\
 &b(y) +_2 \# &\to& b(y),&
 d_2(b(x)) &\to& b(b(d_2(x))),&
 q_2(b(x)) &\to& q_2(x) +_2  b(d_2(x)),\\
& x+_2 b(y) &\to& b(x+_2 y) \;
 \}. &
\end{array}
\]

\begin{lem}
The term $q_2^{p(\ell)}(bb(\#))$ rewrites with $R_2$ to the term
$b^{2^{2^{p(\ell)}}}(\#)$.
\end{lem}

The double-exponentially large counter can be generated by the following copy of \RHL:
\[
\begin{array}{l@{}r@{\;}c@{\;}l@{}r@{\;}c@{\;}l@{\quad}r@{\;}c@{\;}l}
R_3 :=
\{& \# +_3 \# &\to&  \#, &
d_3(\#) &\to& \#,&
q_3(\#) &\to& \#, \\
&f(y) +_3 \# &\to& f(y), &
d_3(f(x)) &\to& f(f(d_3(x))),&
q_3(f(x)) &\to& q_3(x) +_3  f(d_3(x)),\\
&x+_3 f(y) &\to& f(x+_3 y)
\; \}. &
\end{array}
\]

\begin{lem}
The term $q_3^{p(\ell)}(ff(\#))$ rewrites with $R_3$ to the term
$f^{2^{2^{p(\ell)}}}(\#)$.
\end{lem}

We now add to these three TRSs the system $\RM$, which can simulate \M and then make the term small
in case the accepting state \acc is reached. For the following lemma we assume, as before, that \acc is the only
accepting state. In addition, we assume without loss of generality that the initial state $q_0$ is not reachable, i.e., as soon as the machine has made
a transition, it is in a state different from $q_0$ and cannot reach state $q_0$ again.

\begin{lem}\label{hardness:2exp:tech:lem}
The term $t(w)$ can be rewritten with $\RM\cup R_1 \cup R_2 \cup R_3$ to a term of size $1$ iff \M accepts the word $w$.
\end{lem}

\begin{proof}
First, assume that \M accepts the word $w$. Then there is a run of \M on input $w$ such that the accepting state \acc is reached.
We can simulate this run, starting with $t(w)$ by first using $R_1 \cup R_2 \cup R_3$ to generate the term
\[
b^{2^{2^{p(\ell)}}}(q_0(a_1\ldots a_\ell b^{2^{2^{p(\ell)}}}(\#),f^{2^{2^{p(\ell)}}}(\#))).
\]
Since the tape and counter generated this way are large enough, $\RM$ can then simulate the accepting run of \M, and the 
last two rules of $\RM$ can be used to generate the term $\#$, which has size $1$.

For the other direction, we first note that a term of size $1$ can only be reached from $t(w)$ using $\RM\cup R_1 \cup R_2 \cup R_3$
if a term is reached that contains \acc. This function symbol can only be generated by performing transitions of \M, starting with the
input $w$. In fact, while the simulation of \M can start before the system $R_1 \cup R_2 \cup R_3$ has generated
the tape and the counter in full size, rules of \RM can only be applied if the TM locally sees a legal tape configuration.
This means that blanks generated by $R_1$ and $R_2$ can be used even if the application of these systems has not terminated yet.
But if one of the auxiliary symbols employed by these systems is encountered, then no rule simulating a transition of \M is applicable.
These systems cannot generate tape symbols other than blanks, and these blanks are also available to \M in its run.
Thus, $\RM\cup R_1 \cup R_2 \cup R_3$ can only generate a term containing \acc if there is a run of \M on input $w$ that reaches~\acc.
\end{proof}

Thus, by \Cref{hardness:2exp:tech:lem} there is a polynomial-time reduction of the word problem for \M to the small term reachability problem for $\RM\cup R_1 \cup R_2 \cup R_3$
with bound $n = 1$. It remains to construct an appropriate polynomial order for this TRS.

\begin{lem}\label{po:constr:lem}
Termination of $\RM\cup R_1 \cup R_2 \cup R_3$ can be shown 
using a \scomp polynomial order.
\end{lem}

\begin{proof}
Termination of $\RM\cup R_1 \cup R_2 \cup R_3$ can be shown using the following \scomp
polynomial interpretation of the function symbols, which is similar to the interpretation
in \Cref{Ex:PolynomialOrders}:\footnote{%
The definition of this interpretation slightly differs from the one in \cite{FSCD2024} to
ensure that the interpretation and also its extension that deals with confluent systems
(see the proof of \Cref{conf:hard:cor} below) are \scomp.
}
\begin{itemize}
\item
$a(x)$ is mapped to  $x + 3$, for all tape symbols $a$ of the TM (where $a$ can also be
  the blank symbol $b$),
\item
$\#$ is mapped to $4$,
\item
$q(x,y)$ is mapped to $x + y + 3$, for all states $q$ of the TM, in particular also for $q_0$ and $\acc$,
\item
$f(x)$ is mapped to $x + 3$,
\item
$+_1(x,y)$, $+_2(x,y)$, and $+_3(x,y)$ are mapped to  $x + 2y + 1$,
\item
$d_1(x)$, $d_2(x)$, and $d_3(x)$ are mapped to  $3x+1$,
\item
$q_1(x)$, $q_2(x)$, and $q_3(x)$ are mapped to $3 x^2+3x+1$.
\end{itemize}
Obviously, this interpretation satisfies the requirements for \scompa since $\#$ is mapped
to $4 \geq 2$.
 
It remains to show that the polynomial order $\succ$ induced by this polynomial interpretation satisfies $g\succ d$ for all rules 
$g\to d$ of $\RM\cup R_1 \cup R_2 \cup R_3$. 
First, we consider $\RM$:
\begin{itemize} 
\item
for the rule $\acc(x,y) \to \#$, the  left-hand side is mapped to $x + y + 3$, and the right-hand side to $4$, which is smaller than
$x + y + 3$ for all instantiations of $x, y$ with numbers $> 0$,
\item
for rules of the form $a(\acc(x,y)) \to \acc(x,y)$, the  left-hand side is mapped to $x + y + 6$, and the right-hand side to $x + y + 3$,
\item
for all rules of $\RM$ of the form $q(a(x),f(y))\rightarrow a'(q'(x,y))$, 
the left-hand side is mapped to $(x+3) + (y+3) + 3 = x+y+9$, and 
the right-hand side is mapped to $(x+y+3)+3 = x+y+6$,
\item
for all rules of $\RM$ of the form $c(q(a(x),f(y)))\rightarrow q'(ca'(x),y)$,
the left-hand side is mapped to $((x+3) + (y+3) + 3)+3 = x+y+12$, and
the right-hand side is mapped to $((x+3)+3)+y + 3 = x+y+9$.
\end{itemize}
Next, we consider $R_1$:
\begin{itemize}
\item
for the rule $q_0(y_1,y_2) +_1 q_0(z_1,z_2)\to q_0(y_1,y_2)$ of $R_1$, the left-hand side
is mapped to $y_1 + y_2 + 3 + 2 (z_1 + z_2 + 3)+1 = y_1 + y_2 +2z_1 + 2z_2+10$,
and the right-hand side to $y_1 + y_2 + 3$,
\item
for the rule $b(y) +_1 q_0(z_1,z_2)\to b(y)$ of $R_1$, the left-hand side
is mapped to $y + 3 + 2 (z_1 + z_2 + 3)+1 = y +2z_1 + 2z_2+10$,
and the right-hand side to $y+3$,
\item
for the rule $x+_1 b(y) \to b(x+_1 y)$ of $R_1$, the left-hand side is mapped to $x + 2(y + 3)+1=x+2y+7$, and the right-hand side to 
$(x + 2y+1)+3 = x+2y+4$,
\item
for the rule $d_1(q_0(z_1,z_2)) \to q_0(z_1,z_2)$ of $R_1$, the left-hand side is mapped to $3(z_1 + z_2 + 3)+1 = 3z_1 + 3z_2 +10$, 
and the right-hand side to $z_1 + z_2 + 3$,
\item
for the rule $d_1(b(x)) \to b(b(d_1(x)))$ of $R_1$, the left-hand side is mapped to $3(x + 3)+1 = 3x+10$, 
and the right-hand side to $(3x+1+3)+3 = 3x+7$,
\item
for the rule $q_1(q_0(z_1,z_2)) \to q_0(z_1,z_2)$ of $R_1$, the left-hand side is mapped to $3(z_1 + z_2 + 3)^2 + 3(z_1 + z_2 + 3) + 1$,
and the right-hand side to $z_1 + z_2 + 3$,
\item
for the rule $q_1(b(x)) \to q_1(x) +_1  b(d_1(x))$ of $R_1$, the left-hand side is mapped to 
$3(x+3)^2 + 3(x+3) + 1= 3x^2 + 21 x + 37$,
and the right-hand side to $3 x^2 + 3x + 1 + 2(3x+1+3)+1 =   3 x^2 + 9 x + 10$, as in \Cref{Ex:PolynomialOrders}.
\end{itemize}
The rules of $R_2$ and $R_3$ can be treated in a similar way.
\end{proof}

If we use a TM deciding an N2ExpTime-complete problem in this reduction, then the resulting TRS has an N2ExpTime-complete
small term reachability problem. Note that the reduction can be done in polynomial
time
since the size of the term $t(w)$ is polynomial
in the length $\ell$ of the input word $w$, and the number $n = 1$ has constant size both for unary and binary encoding of numbers.
Combining this observation with the upper bound shown before, we obtain the following
complexity result for the small term reachability problem for the class of TRSs considered here.

\begin{thm}\label{completeness:main:thm:poly}
The small term reachability problem is in N2ExpTime for
TRSs whose termination can be shown with a \scomp polynomial order, 
and there are such TRSs for which the small term reachability problem is N2ExpTime-complete. These results hold both for unary and binary encoding
of numbers.
\end{thm}

In the setting considered in this section, restricting the attention to confluent TRSs
does not reduce the complexity. Regarding the upper bound, the argument used in the proof of \Cref{confluence:prop} does not apply
since it is no longer the case that normal forms are of smallest size. Thus, one cannot reduce the complexity from N2ExpTime
to 2ExpTime by only looking at a single rewrite sequence that ends in a normal form. However, our N2ExpTime-hardness
proof does not directly work for confluent TRSs whose termination can be shown with a polynomial order. The reason is
that, for a given nondeterministic Turing machine \M, the rewrite system $\RM\cup R_1 \cup R_2 \cup R_3$ need not be confluent. 
In fact, for a given input word,
there may be terminating runs of the TM that reach the accepting state $\acc$, but also ones that do not reach this state.
Using the former runs, our rewrite system can then generate the term $\#$, whereas this is not possible if we use one of the
latter runs.

We can, however, modify the system $R^\mathcal{M} \cup R_1 \cup R_2 \cup R_3$ such that it becomes
confluent. To this end, we introduce two new function symbols $\#_1$ and $\#_0$ of arity 1 and 0, respectively. 
Moreover, we add the following rules $R_c$:
\[ \begin{array}{rcll}
g(x_1, \ldots, x_n) &\to&  \#_1(\#_0) & \text{for all function symbols $g$ of arity $> 0$
  except $\#_1$},\\
\#_1(\#_1(\#_0)) &\to&  \#_1(\#_0),\\
  \# &\to&  \#_1(\#_0).
\end{array}\]
Clearly,  $R^\mathcal{M} \cup R_1 \cup R_2 \cup R_3 \cup R_c$
is confluent, because any 
term that is not in normal form (i.e., any term
except variables, $\#_0$,  $\#_1(\#_0)$, and terms of the form $\#_1(x)$ for variables $x$)
has the only normal form 
$\#_1(\#_0)$ of size two. (This is the reason why we could not use
a rule like $x +_1 q_0(z_1,z_2) \to x$ in $R_1$, because then
$x +_1 q_0(z_1,z_2)$ would have the two normal forms $x$ and $\#_1(\#_0)$.)
 However,
the term $\#$ of size one is still only reachable from $t(w)$ if the final state of the TM is reached
by a simulation of an accepting computation of $\mathcal{M}$. We extend
the polynomial interpretation in the proof of \Cref{po:constr:lem} as follows:
\begin{itemize}
\item $\#_1(x)$ is mapped to $x + 1$,
\item $\#_0$ is mapped to $2$.
\end{itemize}
Then the polynomial order induced by this polynomial interpretation is \scomp and also orients the rules of $R_c$ from left to right, i.e., termination of the
resulting system can still be shown using a \scomp polynomial order:
\begin{itemize}
\item
  For the additional rules $\#_1(\#_1(\#_0)) \to  \#_1(\#_0)$ and $\# \to \#_1(\#_0)$, the left-hand side is mapped to $4$ and
  the right-hand side to $3$,
\item
  for the rules of the form $g(x_1, \ldots, x_n) \to  \#_1(\#_0)$, the function symbol $g$ can be
  a tape symbol $a$, a state $q$, the symbol $f$, an addition symbol $+_i$, a duplication symbol $d_i$, or a squaring symbol $q_i$.
  The right-hand side of such a rule is always mapped to $3$. For tape symbols and the symbol $f$, the left-hand side
  is mapped to $x_1+3$ and for states to $x_1+x_2+3$. Since the indeterminates are instantiated with natural numbers $>0$, such a left-hand side yields a
  value that is larger than $3$. For the other symbols, we obtain the left-hand sides
  $x_1+2x_2+1$, $3x_1+1$, and $3x_1^2+3x_1+1$. Again, when instantiated
  with natural numbers $> 0$, they yield values that are larger than $3$.
\end{itemize}
 
\begin{cor}\label{conf:hard:cor}
There are confluent TRSs whose termination can be shown with a \scomp polynomial order
for which the small term reachability problem is N2ExpTime-complete, both for unary and binary encoding of numbers.
\end{cor}

As shown in \cite{DBLP:conf/rta/HofbauerL89}, if termination of a TRS can be shown with a \emph{linear} polynomial order 
(i.e., where all polynomials have degree at most $1$), then this implies an exponential bound on the lengths of reduction sequences. 
Again, this bound is tight and one can use the example showing this to obtain a TRS 
that generates an exponentially large tape and an exponentially large counter, similarly to what we have done in the general case.

\begin{exa}
Let $R_d$ consist of just the two $d$-rules from \Cref{ho:la:2ex}. Then the term $d^\ell(s(0))$ can be reduced to $s^{2^\ell}(0)$. 
\end{exa}

\begin{cor}\label{Hardness Linear Polynomials}
The small term reachability problem is in NExpTime for TRSs whose termination can be shown with a \scomp and linear polynomial order.
There are confluent TRSs whose termination can be shown with a \scomp and linear polynomial order for which the small term reachability problem is NExpTime-complete.
These results hold both for unary and binary encoding of numbers.
\end{cor}

\begin{proof}
The upper bound can be shown as before, i.e., one just needs to guess a reduction sequence of at most exponential length starting with $s$,
and then compare the size of the last term with $n$. 
Similarly to the proof of \Cref{size:bound:lem}, we can show that the sizes of the terms in this sequence are exponentially bounded by the size of $s$.
Thus, the next term in the sequence can always be guessed using an NExpTime-procedure and the final comparison takes exponential time.

For the lower bound, we proceed as in the proof of N2ExpTime-hardness for the case of general polynomial orders. Thus,
we assume that $\mathcal{M}$ is an exponentially time bounded nondeterministic TM whose time-bound is $2^{p(\ell)}$ for a polynomial $p$, 
where $\ell$ is the length of the input word. Given such an input word $w = a_1 \ldots a_\ell$ for $\mathcal{M}$, we now construct the term
\[ t'(w) = d_1^{p(\ell)}(b(q_0(a_1 \ldots a_\ell d_2^{p(\ell)}(b(\#)), d_3^{p(\ell)}(f(\#)))).\]
Instead of $R_1, R_2, R_3$, we now only need their rules for $d_1$, $d_2$, and $d_3$; let $R_d'$ denote this system of 6 rules. 
As above, we can show that the term $t'(w)$ can be rewritten with $R^\mathcal{M} \cup R_d'$ to a term of size 1 iff 
$\mathcal{M}$ accepts the word $w$. Moreover, termination of $R^\mathcal{M} \cup R_d'$ can be proved by the \scomp and linear polynomial order obtained 
from the one in the proof of \Cref{po:constr:lem} by removing the (non-linear) interpretations of $q_1, q_2, q_3$.

Similarly to the proof of \Cref{conf:hard:cor}, we can prove that NExpTime-hardness also
holds for a \emph{confluent} TRS whose termination can be 
shown with a size compatible and linear polynomial order. The reason is that termination of the modified confluent TRS $R^\mathcal{M} \cup R_d' \cup R_c$ can be shown by 
the \scomp and linear polynomial order that results from the one employed in the proof of \Cref{conf:hard:cor} by removing the (non-linear) interpretations of $q_1, q_2, q_3$.
\end{proof}

\section{Term rewriting systems shown terminating with\texorpdfstring{\\}{} a Knuth-Bendix order without special symbol}\label{kbo:sect}

Without any restriction, there is no primitive recursive bound on the length of derivation sequences for TRSs whose termination can be shown 
using a Knuth-Bendix order~\cite{DBLP:conf/rta/HofbauerL89}, but a uniform multiple recursive upper bound is shown in~\cite{DBLP:journals/iandc/Hofbauer03}. 
Here, we restrict the attention to KBOs without a \emph{special symbol}, i.e., without a unary symbol of weight zero. For such KBOs, an
exponential upper bound on the derivation length was shown in~\cite{DBLP:conf/rta/HofbauerL89}.\footnote{%
Actually, this result was shown in~\cite{DBLP:conf/rta/HofbauerL89} only for KBOs using
weights in \N, but it also holds for KBOs with non-negative weights in \R.
This is an easy consequence of our \Cref{KBO:size:bound}.
}
Given the results proven in the previous section,
one could now conjecture that in this case the small term reachability problem is NExpTime-complete. However, we will show below that the complexity is actually
only PSpace. In fact, the TRSs yielding the lower bounds for the derivation length considered in the previous section have not only long reduction sequences
(of double-exponential or exponential length),
but are also able to produce large terms (of double-exponential or exponential size). For KBOs without special symbol, this is not the case.
The following lemma provides us with a linear bound on the sizes of reachable terms.
It will allow us to show a \emph{PSpace upper bound} for the small term reachability problem.

\begin{lem}\label{KBO:size:bound}
Let $R$ be a TRS whose termination can be shown using a KBO without special symbol, and $s_0, s_1$ terms such that
$
s_0 \too_R s_1.
$
Then the size of $s_1$ is linearly bounded by the size of $s_0$.
In other words, for every such TRS $R$, there exists some number $c \geq 0$ such that 
$|s_1|\leq c \cdot |s_0|$ holds for all terms
$s_0, s_1$ with $s_0 \too_R s_1$. 
\end{lem}

\begin{proof}
Fix a KBO with weight function $w$ showing termination of $R$ such that all symbols of
arity $1$ have weight $> 0$.
We define 
\[
\wnull := \min\{w(f) \mid w(f)> 0\ \mbox{and $f$ is a function symbol in $R$ or a variable\footnotemark}\},
\]
and\footnotetext{%
Recall that all variables have the same weight $w_0>0$.} let $\wmax$ be the maximal weight of a function symbol in $R$ or a variable.
As the weights of function symbols not occurring in $R$ have no influence on the orientation of the rules in $R$ with the given KBO,
we can assume without loss of generality that their weight is $\wnull$.

Let $t$ be a term and $n_i(t)$ for $i=0,\ldots,k$ the number of occurrences of symbols of arity $i$ in $t$, 
where $k$ is the maximal arity of a symbol occurring in $t$.\footnote{Variables have arity $0$.}
Note that $|t| = n_0(t) + n_1(t) + \ldots + n_k(t)$.
The following fact, which can easily be shown by induction on the structure of $t$, is stated in~\cite{KnuthBendix}:
\[
n_0(t) + n_1(t) + \ldots + n_k(t) = 1 + 1 \cdot n_1(t) + 2 \cdot n_2(t) + \ldots + k \cdot n_k(t).
\]
In particular, this implies that $n_0(t)\geq n_2(t) + \ldots + n_k(t)$. 
Since symbols of arity $0$ and $1$ have weights $> 0$, we know that 
\[
w(t) \geq \wnull \cdot (n_0(t) + n_1(t)) \geq \wnull \cdot n_0(t) \geq \wnull \cdot (n_2(t) + \ldots + n_k(t)).
\]
Consequently, $2  \cdot  \wnull^{-1} \cdot w(t) \geq n_0(t) + n_1(t) + \ldots + n_k(t) = |t|$. This shows that the size of a term is linearly bounded
by its weight. Conversely, it is easy to see that the weight of a term is linearly bounded by its size: $w(t)\leq \wmax \cdot |t|$.

Now, assume that $s_0 \too_R s_1$. Since termination of $R$ is shown with our given KBO, we know that $w(s_0)\geq w(s_1)$, and thus
$\wmax \cdot |s_0| \geq w(s_1) \geq 1/2  \cdot  \wnull \cdot |s_1|$. 
This yields $|s_1| \leq 2 \cdot \wnull^{-1} \cdot \wmax \cdot |s_0|$.
Since we have assumed that the TRS $R$ is fixed, the number $2 \cdot \wnull^{-1} \cdot \wmax$ is a constant.
\end{proof}

In particular, this means that the terms encountered during a rewriting sequence starting with a term $s$ can each be stored using only polynomial
space in the size of $s$. Given that the length of such a sequence is exponentially bounded, we can decide the small term reachability problem
by the following nondeterministic algorithm: 
\begin{itemize}
\item
  guess a rewrite sequence $s\to_R s_1 \to_R s_2 \to_R \ldots$ and always store only the current term;
\item
  in each step, check whether $|s_i|\leq n$ holds. If the answer is ``yes'' then stop and accept. Otherwise, guess the next rewriting step;
  if this is not possible since $s_i$ is irreducible, then stop and reject. 
\end{itemize}
This algorithm needs only polynomial space since, by \Cref{KBO:size:bound}, the size of each term $s_i$ is linearly bounded by the size of $s$.
It always terminates since $R$ is terminating. If there is a term of size $\leq n$ reachable from $s$, then the algorithm is able
to guess the sequence leading to it, and thus it has an accepting run.
Otherwise, all runs are terminating and rejecting.
Note that, by Savitch's theorem~\cite{DBLP:journals/jcss/Savitch70},
NPSpace is equal to PSpace.
Thus, we obtain the following complexity upper bound.

\begin{lem}
The small term reachability problem is in PSpace for TRSs whose termination can be shown with a KBO without special symbol,
both for unary and binary encoding of numbers.
\end{lem}

It remains to prove the corresponding \emph{lower bound}.
Let $\M$ be a polynomial space bounded TM, and $p$ the polynomial that yields the space bound. Then there is a polynomial $q$ such that any run of
\M longer than $2^{q(\ell)}$ on an input word $w$  of length $\ell$ is cyclic. Thus, to check whether \M accepts $w$, it is sufficient to consider only runs
of length at most $2^{q(\ell)}$. However, in contrast to the reduction used in the previous section, we cannot generate an exponentially large unary
down counter using a TRS whose termination can be shown with a KBO without special symbol. Instead, we use a polynomially large \emph{binary} down counter that
is decremented, starting with the binary representation $10^{q(\ell)}$ of $2^{q(\ell)}$ (see \Cref{down:cont:ex}). 
For example, if $q(\ell) = 3$, then we represent the number
$2^{q(\ell)} = 2^3 = 8$ as the binary number $10^{q(\ell)} = 1000$. The construction of the TRS \RMbin simulating \M given below is
very similar to the construction given in the proof of Lemma~7 in~\cite{DBLP:conf/lpar/BonfanteM10}.

As signature for \RMbin we again use the tape symbols of \M as unary function symbols, but now also the states are treated as unary symbols. In addition,
we need the unary function symbols $0$ and $1$ to represent the counter, as well as primed versions $a', q', 1'$ of the tape symbols $a$, the states $q$, and
the symbol $1$.
For a given input word $w = a_1\ldots a_\ell$ of \M, we now construct a term that starts with the binary representation of $2^{q(\ell)}$ and is
followed by enough tape space for a $p(\ell)$ space bounded TM to work on:
\[
\tilde{t}(w) := 10^{q(\ell)}(b^{p(\ell)}(q_0 (a_1\ldots a_\ell(b^{p(\ell)-\ell}(\#))))).
\]
Clearly, $\tilde{t}(w)$ can be constructed in polynomial time.

The TRS \RMbin is now defined as follows. The first part decrements the counter (as in \Cref{down:cont:ex}) and by
doing so ``sends a prime'' to the right: 
\[
\begin{array}{lcll}
1(a(x)) \to 0(a'(x)) &\mbox{\ and\ }& 0(a(x)) \to 1'(a'(x)) &\mbox{for all  tape symbols $a$,}\\[.2em]
0(1'(x)) \to 1'(1(x)), &&
1(1'(x)) \to 0(1(x)). &
\end{array}
\]
The prime can go to the right on the tape until it reaches a state, which it then turns into its primed version:
\[
\begin{array}{ll}
a'g(x) \to ag'(x) & \mbox{\ \ for tape symbols $a$ and tape symbols or states $g$}.
\end{array}
\]
Only primed states can perform a transition of the TM:
\[
\begin{array}{ll}
 q_1'(a_1(x)) \to a_2(q_2(x)) & \mbox{\ \ for each transition $(q_1,a_1,q_2,a_2,r)$ of \M},\\[.2em]
 c(q_1'(a_1(x)) \to q_2(c(a_2(x))) & \mbox{\ \ for each transition $(q_1,a_1,q_2,a_2,l)$ of \M}\\
                                   & \mbox{\ \  and tape symbol $c$.}
\end{array}
\]
Again, the blank symbol $b$ is also considered as a tape symbol of $\M$.
Note that the rôle of the counter is not to restrict the number of transition steps simulated by \RMbin. Instead it
produces enough primes to allow the simulation of at least $2^{q(\ell)}$ steps, while termination can still be shown
using a KBO without special symbol.

Once the unique final accepting state $\acc$ is reached, we remove all symbols other than $\#$:
\[
\begin{array}{ll}
a(\acc(x)) \to \acc(x) & \mbox{\ \ where $a$ is a tape symbol or $0$ or $1$},\\[.2em]
\acc(x) \to \#.
\end{array}
\]

\begin{lem}
The term $\tilde{t}(w)$ can be rewritten with \RMbin to a term of size $1$ iff \M accepts the word $w$.
\end{lem}

\begin{proof}
If \M accepts the word $w$, then there is a run of \M on input $w$ that ends in the state $\acc$, uses at most $p(\ell)$ space, and
requires at most $2^{q(\ell)}$ steps. This run can be simulated by \RMbin by decrementing the counter, sending a prime to the state, applying
a transition, decrementing the counter, etc. Since the counter can be decremented $2^{q(\ell)}$ times, we can use this approach to simulate a run
of length at most $2^{q(\ell)}$. Once the accepting state is reached, we can use the last two rules to reach the term $\#$, which has size $1$.

Conversely, we can only reach a term of size one, if these cancellation rules are applied. This is only possible if first the accepting state
has been reached by simulating an accepting run of \M.
\end{proof}

To conclude our proof of the lower complexity bound, it remains to construct an appropriate KBO for \RMbin.

\begin{lem}\label{Termination RMbin}
Termination of \RMbin can be shown with a KBO without special symbol.
\end{lem}

\begin{proof}
It is easy to see that the KBO that assigns weight $1$ to all function symbols and to all variables, and uses the precedence
order $1 > 0 > 1'$ and $q' > a' > a > q$ for states $q$ and tape symbols $a$, orients all rules of \RMbin from left to right.\footnote{%
This KBO is similar to the one introduced in Example~10 of~\cite{DBLP:conf/lpar/BonfanteM10}.
}
\end{proof}

As in the case of the TRSs considered in the previous section, confluence does not reduce the complexity 
of the small term reachability problem for TRSs shown terminating with a KBO
without special symbol. In fact, we can again extend the TRS \RMbin such that it becomes confluent. To this purpose, we add two new function symbols
$\#_1$ and $\#_0$ of respective arity $1$ and $0$, and two new rules:
\[
\begin{array}{ll}
g(x) \to \#_1(\#_0) & \mbox{\ \ for all unary function symbols $g$ different from $\#_1$},\\[.2em]
\# \to \#_1(\#_0).
\end{array}
\]
With this addition, every non-variable term built using the original signature of \RMbin can be reduced to $\#_1(\#_0)$, which proves confluence.
To show termination of the extended TRS, we modify and extend the KBO from the proof of
\Cref{Termination RMbin} as follows. All function symbols in the original signature of \RMbin (including $\#$)
now get weight $2$, and the symbols $\#_1$ and $\#_0$ as well as the variables get weight $1$. The precedence order is extended by setting 
$\# > \#_1$.
It is easy to see that the KBO defined this way shows that the extended TRS is terminating.

Combining the results obtained in this section, we thus have determined the exact complexity of the small term reachability problem
for our class of TRSs.

\begin{thm}\label{kbo:thm} 
The small term reachability problem is in PSpace for TRSs whose termination can be shown with a KBO without special symbol,
and there are confluent such TRSs for which the small term reachability problem is PSpace-complete.
These results hold both for unary and binary encoding of numbers.
\end{thm}

\section{Variants of the small term reachability problem}\label{variants:sect}

In order to investigate how much our complexity results depend on the exact formulation
of the condition on the reachable terms, we now consider some variants of the problem.
First, we introduce two variants for which the complexity is the same as for the small term reachability problem.
Then, we investigate a variant that partially leads to different complexity results, depending on the encoding of numbers.

\paragraph*{Two ``harmless'' variants}

In these variants we impose conditions on the reachable terms that are not related to size. 

\begin{defi}\label{harmless}
Let $R$ be a TRS.
\begin{enumerate}
\item
Given a term $s$ and a function symbol $f$,
the \emph{symbol reachability problem} for $R$ asks whether there exists a term $t$
such that $s\too_R t$ and $f$ occurs in $t$. 
\item
Given two terms $s$ and $t$,
the \emph{term reachability problem} for $R$ asks whether $s\too_R t$.
\end{enumerate}
\end{defi}
 
For the classes of reduction orders considered in the previous three sections, we obtain the same
complexity results as in the case of the small term reachability problem. For proving the
\emph{upper bounds}, it is 
enough to observe that the new test applied to the final term in the guessed reduction
sequence  (occurrence of $f$ or syntactic equality with the term $t$)
is still possible within the respective time-bound.

Regarding the \emph{lower bounds}, we can basically use the same reductions as in the previous three sections. 
For the symbol reachability problem, we can, however, dispense with the rules that make the terms smaller once the final state $\acc$ of the TM is reached, 
and just use this state as the function symbol to be reached.
For the term reachability problem, note that in these reductions the
fixed term $\#$ can be reached iff
the TM can reach its final state.

\begin{cor}
The complexity results stated in \Cref{length:red:sect,po:sect,kbo:sect} for the small term reachability problem also hold for the
symbol reachability problem and the term reachability problem.
\end{cor}

\paragraph*{Large term reachability}

We now investigate the following dual problem to small term reachability, where the comparison $|t| \leq n$ is replaced with $|t| \geq n$.
Although at first sight this may look like an innocuous change, it turns out that it has a considerable impact on the complexity of the problem
in some cases.

\begin{defi}
Let $R$ be a TRS. Given a term $s$ and a natural number $n$,
the \emph{large term reachability problem} for $R$ asks whether there exists a term $t$
such that $s\too_R t$ and $|t|\geq n$.
\end{defi}

For this problem, we can prove complexity upper bounds for all terminating TRSs without having to make any assumption on how termination is shown.
However, these upper bounds depend on how the number $n$ in the formulation of the problem is assumed to be encoded. The main idea is
that the value of the number $n$ (rather than the size of its encoding) provides us with a polynomial upper bound on the sizes of the terms to be considered.
 
\begin{thm}\label{general:upper:thm}
For terminating TRSs, the large term reachability problem is in ExpSpace (PSpace), if binary (unary) encoding of numbers is assumed.
\end{thm}

\begin{proof}
Let $R$ be a terminating TRS, let $s$ be a term, and $n$ be a natural number (in unary or binary representation).
We first check whether $|s|\geq n$. If this is the case, then our procedure terminates and returns ``success''.
By adapting the algorithm sketched in the proof of \Cref{lenred:upper:bound}, it is easy to see that this test can be performed in linear time
in the combined size of $s$ and $n$, both for unary and binary encoding of $n$.
Otherwise, we guess a rewrite sequence $s\to_R s_1 \to_R s_2 \to_R \ldots $, of which we keep only the
current term in memory. We stop generating successors
\begin{itemize}
\item
if a term of size $\geq n$ is reached, in which case we return ``success'', or
\item
if no successor exists, in which case we return ``failure''.
\end{itemize}
Since $R$ is terminating, one of these two cases occurs after finitely many rewrite steps.
The sizes of the terms in the sequence, except possibly
the last one, are bounded by $n$. The last term in the sequence may be larger than $n$, but it is obtained by applying a single rewrite
step to a term of size at most $n$. Thus, its size is still polynomial in $n$.
Consequently, if we consider the sizes of the terms in the sequence with respect to the size of the binary (unary) representation of $n$,
then these sizes are exponentially (polynomially) bounded by the size of the representation.
This yields the ExpSpace (PSpace) upper bound stated in the theorem.
\end{proof}

If we take the reduction order used to show termination of $R$ into account, then we can improve on these general upper bounds,
and in some cases also show corresponding lower bounds.

First, we note that for \emph{length-reducing TRSs} $R$, the change from $|t| \leq n$ to $|t| \geq n$ trivializes the problem. 
In fact, in a reduction sequence for $R$, the start term $s$ is always the largest one. Thus, if $R$ is length-reducing, 
then there exists a term $t$ with $s\too_R t$ and $|t|\geq n$ iff $|s| \geq n$. 

\begin{thm}\label{large:length-reducing}
For length-reducing TRSs, the large term reachability problem can be decided in
(deterministic) linear time, both for unary and binary encoding of numbers.
\end{thm}

For the case of a \emph{KBO without special symbol}, a PSpace upper bound can not only be shown for unary, but also for binary encoding of numbers,
using the same approach as for the small term reachability problem. 
Concerning the lower bound, we can still construct the term $\tilde{t}(w)$, whose size does not change during the run of the TM. 
When reaching the final state $\acc$, the size of the term must be increased by one to make it larger than the terms encountered so far in the
reduction sequence. Thus, instead of the rules 
$a(\acc(x)) \to \acc(x)$ and 
$\acc(x) \to \#$, we now use the rule
$\acc(x) \to \tilde{q}(\tilde{q}(x))$ for a new function symbol $\tilde{q}$. With respect to this modified TRS, we then know that $\tilde{t}(w)$
reduces to a term of size $|\tilde{t}(w)| + 1$ iff the TM \M accepts the word $w$. 
Note that the number $|\tilde{t}(w)| + 1$ has a representation that is polynomial in the length of $w$, both for unary and binary encoding of numbers. 
For the KBO showing termination of the modified TRS, $\tilde{q}$ and the variables
get the weight 1 and all other function symbols get the weight $2$. 
The precedence order for the old symbols stays as before, and the new symbol $\tilde{q}$
is defined to be smaller than $\acc$.

The hardness result also holds for \emph{confluent} TRSs whose termination can be shown
with a KBO without special symbol. This can be proved by extending the modified TRS by the
rules
$g(x) \to \#$ for all unary function symbols $g$ such that it becomes confluent. To show
its termination, the KBO can be extended such that $\#$ gets the weight 1.
Overall, we thus obtain the following complexity result.

\begin{thm}\label{largeTerm-KBO}
The large term reachability problem is in PSpace for TRSs whose termination can be shown with a KBO without special symbol,
and there are confluent such TRSs for which the large term reachability problem is PSpace-complete.
These results hold both for unary and binary encoding of numbers.
\end{thm}

For the case of a \scomp and \emph{linear polynomial order}, a NExpTime upper bound can be shown 
(both for binary and unary encoding of numbers) as for the small term reachability problem. 
This improves on the general ExpSpace upper bound provided by \Cref{general:upper:thm} for binary encoding,
but is worse than the PSpace upper bound provided by \Cref{general:upper:thm} for unary encoding.
To  prove the corresponding \emph{lower bound for binary encoding} of numbers, we modify the proof of
\Cref{Hardness Linear Polynomials} as follows. 
Let $\mathcal{M}$ be an exponentially time bounded nondeterministic TM whose time-bound is $2^{p(\ell)}$ for a polynomial $p$, 
where $\ell$ is the length of the input word. Given such an input word $w = a_1 \ldots a_\ell$ for $\mathcal{M}$, we now construct the term
\[ t''(w) = d_1^{p(\ell)}(b(q_0(a_1 \ldots a_\ell d_2^{p(\ell)}(b(\#)), \;
d_3^{p(\ell)}(f(\#)), \; d_4'^{p(\ell) + 1}(f(\#))))).\]
Hence, we employ a new unary function symbol $d_4'$. 
Moreover, all symbols for states (like $q_0$) are now assumed to  have arity $3$ instead of $2$. 
The idea underlying the use of the additional third argument 
$d_4'^{p(\ell) + 1}(f(\#))$
is to compensate the decrease of the counter generated from $d_3^{p(\ell)}(f(\#))$. This counter term
evaluates to $f^{2^{p(\ell)}}(\#)$ and then decreases the number of $f$'s in each
evaluation step of the TM. To compensate this decrease, as soon as one reaches the final
state $\acc$, the symbol $d_4'$ is turned into $d_4$ and the third argument 
$d_4'^{p(\ell) + 1}(f(\#))$ evaluates to $f^{2^{p(\ell)+1}+2}(\#)$.

Since the state symbols now have arity 3, the non-recursive $d_1$-rule from $R_1$ must be changed to
\[ d_1(q_0(z_1,z_2,z_3)) \to q_0(z_1,z_2,z_3). \]
The rules for $d_4$ are similar to the ones for $d_3$, but in the end, the size of the
result is increased by $2$:
\[d_4(\#) \to f(f(\#)), \qquad d_4(f(x)) \to f(f(d_4(x))).\]
Moreover, the TRS $\RM$ that simulates \M now has the following modified rewrite rules:
\begin{itemize}
\item
  For each transition $(q,a,q',a',r)$ of \M it has the rule\ \ $q(a(x),f(y),z)\rightarrow
  a'(q'(x,y,z))$.
\item
   For each transition $(q,a,q',a',l)$ of \M it has the rule\ \ $c(q(a(x),f(y),z))\rightarrow q'(ca'(x),y,z)$
   for every tape symbol $c$ of \M.
\end{itemize}
The rules for $\acc$ are also modified. Whenever this final state is reached, the
outermost $d_4'$ in the third argument is turned into $d_4$:
\[ \acc(x,y,d_4'(z)) \to \acc(x,y,d_4(z)).\]
Moreover, all $d_4'$-symbols below $d_4$ are also turned into $d_4$:
\[ d_4(d_4'(z)) \to d_4(d_4(z)).\]
Consequently, if \M accepts the word $w$, then we obtain the following reduction sequence:
\[ \begin{array}{rcl}
t''(w) & = & d_1^{p(\ell)}(b(q_0(a_1 \ldots a_\ell d_2^{p(\ell)}(b(\#)), \;
d_3^{p(\ell)}(f(\#)), \; d_4'^{p(\ell) + 1}(f(\#)))))\\
&\to^*&  b^{2^{p(\ell)}}(q_0(a_1 \ldots a_\ell b^{2^{p(\ell)}}(\#), \; 
f^{2^{p(\ell)}}(\#), \;
d_4'^{p(\ell) + 1}(f(\#))))\\
& = & \bar{t}.
\end{array}
\]
Note that all further evaluations of $\bar{t}$ until reaching $\acc$ cannot
increase the size of the term anymore. The reason is that the length of the tape
(represented by the context around $q_0$ and $q_0$'s first argument) always remains
$2^{p(\ell)} + \ell + 2^{p(\ell)} = 2 \cdot 2^{p(\ell)} + \ell$, and the size of the counter $f^{2^{p(\ell)}}(\#)$ is
decremented in each simulated evaluation step of the TM. Of course, one could also start
with evaluation steps of the TM before evaluating $d_1$, $d_2$, and $d_3$, but this would
not yield a term of larger size than $\bar{t}$. So $|\bar{t}| = 3 \cdot 2^{p(\ell)}
+ \ell + p(\ell) + 6$ is the largest size
of any term during the reduction until reaching $\acc$.
In particular, if \M does not accept the word $w$, then one cannot reach any term of
size larger than $|\bar{t}|$.

However,  if \M accepts the word $w$, then we can reach a term that is larger than
$|\bar{t}|$. In fact, in this case there is a reduction sequence
\[
  t''(w) \; \to^* \; \bar{t}
  \; \to^* \; C[\acc(t_1, t_2, d_4'^{p(\ell) + 1}(f(\#)))]
  \]
  for a context $C$ and two terms $t_1$ and $t_2$. If the hole of the context $C$ does 
not count for the size, then we have $|C| + |t_1| = 2 \cdot 2^{p(\ell)} + \ell + 1$, i.e., it
corresponds to the length of the tape of the TM plus the $\#$ in $t_1$. Moreover, we know that $|t_2| \geq 1$, because
this term corresponds to the counter term which may have been decreased to $\#$  (but
which may also still be larger). Now we can continue the reduction as follows:
\[ \begin{array}{rcl}
  C[\acc(t_1, t_2, d_4'^{p(\ell) + 1}(f(\#)))] &\to& 
  C[\acc(t_1, t_2, d_4(d_4'^{p(\ell)}(f(\#))))] \\
   &\to^*& 
  C[\acc(t_1, t_2, d_4^{p(\ell)+1}(f(\#)))] \\
   &\to^*& 
  C[\acc(t_1, t_2, f^{2^{p(\ell)+1}}(d_4(\#)))] \\
  &\to& 
  C[\acc(t_1, t_2, f^{2^{p(\ell)+1}+2}(\#))].
\end{array}\]
Hence, the size of the resulting term is
$|C| + 1 + |t_1| + |t_2| + 2^{p(\ell)+1}+2 + 1
\geq
2 \cdot 2^{p(\ell)} + \ell + 1 + 1 + 
1 + 2^{p(\ell)+1}+2 + 1 =
4 \cdot 2^{p(\ell)}+ \ell+ 6
>
3 \cdot 2^{p(\ell)}
+ \ell + p(\ell) + 6
= |\bar{t}|$,
because 
$2^n > n$ holds for every natural number $n$ (and thus, it also holds for $n = p(\ell)$).
Consequently, we have shown the following lemma.

\begin{lem}
The term $t''(w)$ can be rewritten with our modified TRS
to a term of size at least
$4 \cdot 2^{p(\ell)}+ \ell+ 6$
 iff 
 $\mathcal{M}$ accepts the word $w$.
\end{lem}
Note that the size of the binary representation of the number $4 \cdot 2^{p(\ell)}+ \ell+ 6$ is polynomial in the length $\ell$ of the input word $w$.

It remains to show that termination of our modified TRS
can still be shown with a linear and \scomp polynomial ordering. To this end,  we use the polynomial
interpretation from the proof of \Cref{Hardness Linear Polynomials}. The only change is that for all
states $q$ of the TM, $q(x,y,z)$ is mapped to $x + y + z + 3$. Moreover, while $d_4(x)$ is
mapped to $3x + 1$, $d_4'(x)$ is mapped to $3x + 2$.

This hardness result again holds for \emph{confluent} TRSs as well. To show this, we use
an analogous extension as for KBO in \Cref{largeTerm-KBO} and add a fresh constant $\#_0$ and
rules $g(x_1,\ldots,x_n) \to \#_0$ for all 
function symbols $g \neq \#_0$ such that the TRS becomes confluent.
The polynomial ordering is extended by mapping $\#_0$ to $2$.

\begin{thm}\label{large:term:lin:po:thm}
The large term reachability problem is in NExpTime for TRSs whose termination can be shown with a \scomp and linear polynomial order, 
even if binary encoding of numbers is assumed.
There are confluent TRSs whose termination can be shown with a \scomp and linear polynomial order for which the large term reachability problem is NExpTime-complete
if binary encoding of numbers is assumed.
\end{thm}
 
Assuming \emph{unary encoding} of numbers lowers the complexity of the large term reachability problem to PSpace
for all terminating TRSs, and thus also for TRSs whose termination can be shown with a linear polynomial order. 
The reason why the reduction showing the NExpTime lower bound for binary encoding of numbers does not work for unary encoding is that
the size of the
unary representation of the number $n = 4 \cdot 2^{p(\ell)}+ \ell+ 6$ is exponential in $\ell$.

For \emph{general polynomial orders}, \Cref{general:upper:thm} yields
  ExpSpace (and PSpace for unary encoding of numbers) upper bounds for the large term reachability problem,
which improves on the N2ExpTime complexity we have shown for the small term reachability
problem.
Note that adapting the N2ExpTime
hardness proof from the small to the large term reachability problem does not even work here for the case of binary encoding of numbers
since in this case the number $n$ would need to be double-exponentially large, and thus its binary representation would still be
of exponential size, which prevents a polynomial-time reduction.

\section{Conclusion}

The results of this paper show that the complexity of the small term reachability problem is closely
related to the derivational complexity of the class of term rewriting systems considered. Interestingly, restricting
the attention to confluent TRSs reduces the complexity only for the class of length-reducing systems, but not for the other two
classes considered in this paper.
For length-reducing TRSs and TRSs shown terminating with a
KBO without special symbol or
with a size compatible linear polynomial order, our hardness results already hold when
only considering \emph{linear} TRSs (where all occurring terms contain every variable at
most once). For the case of the KBO, it even suffices to consider \emph{unary} TRSs
\cite{AAECC08} (where
all function symbols have arity 1 or 0, i.e., these are ``almost'' string rewriting systems).
In contrast,
our hardness result for general size compatible polynomial orders uses a\linebreak non-linear TRS.

The investigations in this paper were restricted to
classes of TRSs defined
by classical reduction orders (restricted forms of KBO and polynomial orders) that yield relatively
low bounds on the derivational complexity of the TRS.
In the future, it would be interesting to consider corresponding classes of TRSs defined
via more recent powerful techniques for termination analysis which also yield similar bounds 
on the derivational complexity, e.g.,   matrix interpretations
\cite{EndrullisWZ08} and their restriction to triangular matrices \cite{MoserSW08}, arctic
matrix interpretations \cite{KoprowskiW09}, 
and match bounds \cite{GeserHWZ07}.

The derivational complexity of
TRSs shown terminating
by KBOs with a unary function symbol of weight zero or by recursive path
orders is much
higher~\cite{DBLP:journals/tcs/Hofbauer92,DBLP:journals/iandc/Hofbauer03,DBLP:journals/tcs/Lepper01,DBLP:journals/aml/Lepper04,DBLP:journals/tcs/Weiermann95}.
From a theoretical point of view, it would be interesting to see whether using
such reduction orders or other  powerful techniques for showing termination like dependency pairs~\cite{DBLP:journals/jar/GieslTSF06}  also result in a very high complexity of the small term reachability problem.
This is not immediately clear since, as we have seen in this paper for the case of KBOs without special symbols, 
the complexity of this problem not only depends on the length of reduction sequences,
but also on whether rewrite sequences can generate large terms.
Another interesting question could be to investigate whether our complexity upper bounds for the small term reachability problem still hold
if the TRS is not assumed to be fixed, or (in the case of polynomial orders) the \scompa restriction is removed.

On the  practical
side, up to now we have only used length-reducing rules to shorten DL proofs. Basically, these
rules are generated by finding frequent proof patterns (currently by hand) and replacing
them by a new ``macro rule''. The results of \Cref{length:red:sect} show that, in this case, confluence of
the rewrite system is helpful. When translating between different proof calculi,
length-reducing systems will probably not be sufficient. Therefore, 
we will investigate with what kinds of techniques proof rewriting systems (e.g.,
translating between different proof calculi for \EL) can be shown terminating.
Are polynomial orders or KBOs without unary function symbol of weight zero sufficient, or are more powerful approaches for showing termination needed?
In this context, it might also be interesting to consider rewriting modulo equational theories~\cite{DBLP:journals/tcs/BachmairD89,DBLP:journals/siamcomp/JouannaudK86}
and associated approaches for showing termination~\cite{DBLP:conf/wrla/AlarconLM10,DBLP:conf/rta/GieslK01,DBLP:journals/tcs/JouannaudM92,DBLP:journals/iandc/Rubio02}.
For example, it makes sense not to distinguish between proof steps that differ only in the order of the
prerequisites. 
Hence, rewriting such proofs could be represented via term rewriting modulo
associativity and commutativity.

We have also looked at variants of the small term reachability problem to investigate how much our complexity results depend on the exact definition
of the condition imposed on the reachable terms. While there are variants (symbol and term reachability) that do not
affect our complexity results, considering the large term reachability problem 
leads to improved complexity upper bounds for length-reducing TRSs (linear), for TRSs shown terminating with a linear polynomial order if unary encoding of numbers is assumed (PSpace),
and TRSs shown terminating with a general polynomial order,
even if binary encoding of number is assumed (ExpSpace). 
The latter two upper bounds (ExpSpace for binary and PSpace for unary encoding)
actually hold for all terminating TRSs and not just for TRSs shown terminating with a (linear) polynomial order.
For the case of general polynomial orders, it is an open problem whether corresponding 
complexity lower
bounds can be proved.
Our main complexity results are summarized in \Cref{Overview} in the introduction.

\bibliographystyle{alphaurl}
\bibliography{literature}

\end{document}